\documentclass[11pt,a4paper,twoside,onecolumn,titlepage,fleqn]{article}
\usepackage[left=2cm, right=2cm, top=2cm, bottom=2cm]{geometry}
\usepackage[english]{babel}
\usepackage[utf8]{inputenc}
\usepackage[T1]{fontenc}
\usepackage{amsmath,amsfonts,amssymb,amsthm}
\usepackage{hyperref}
\usepackage{caption}
\usepackage[all,pdf,rotate]{xy}
\usepackage{tikz}
\usetikzlibrary{shadings}
\usepgflibrary{arrows.meta}
\newtheorem{definition}{Definition}
\newtheorem{theorem}{Theorem}
\newtheorem{lemma}{Lemma}
\newtheorem{remark}{Remark}
\newcommand{\ie}{\textit{i.e.}}
\newcommand{\eg}{\textit{e.g.}}
\newcommand{\ep}{e_P}
\newcommand{\el}{e_L}
\newcommand{\eal}{e_{AL}}
\newcommand{\rset}{\mathbb{R}}
\newcommand{\pru}{P\rset^1}
\newcommand{\prd}{P\rset^2}

\newcommand{\F}{\mathfrak{F}}
\newcommand{\cel}{\mathcal{C}}
\newcommand{\nm}{\mathcal{N}^{-}}
\newcommand{\mem}{\mathcal{E}}
\newcommand{\wl}{\mathcal{W}}
\newcommand{\wlu}{\mathcal{V}}
\newcommand{\tem}{\widetilde{\mathcal{E}}}
\newcommand{\wlt}{\widetilde{\mathcal{W}}}
\newcommand{\hem}{\widehat{\mathcal{E}}}
\newcommand{\wlh}{\widehat{\mathcal{W}}}
\newcommand{\bem}{\overline{\mathcal{E}}}
\newcommand{\wlb}{\overline{\mathcal{W}}}
\newcommand{\sem}{\mathcal{S}}
\newcommand{\sembb}{\mathbb{S}}
\newcommand{\sembul}{\mathcal{S}^\bullet}
\newcommand{\semast}{\mathcal{S}^\ast}
\newcommand{\semdag}{\mathcal{S}^\dagger}
\newcommand{\use}{\mathcal{U}}
\newcommand{\talpha}{\tilde{\alpha}}
\newcommand{\tbeta}{\tilde{\beta}}
\newcommand{\balpha}{\bar{\alpha}}
\newcommand{\bbeta}{\bar{\beta}}
\newcommand{\halpha}{\hat{\alpha}}
\newcommand{\hbeta}{\hat{\beta}}
\newcommand{\te}{\widetilde{E}}
\newcommand{\he}{\widehat{E}}
\newcommand{\be}{\overline{E}}
\newcommand{\tu}{\widetilde{U}}
\newcommand{\bu}{\overline{U}}
\newcommand{\hu}{\widehat{U}}
\newcommand{\ebul}{E^\bullet}
\newcommand{\east}{E^\ast}
\newcommand{\ediam}{E^\diamond}
\newcommand{\ediamd}{E_\diamond}
\newcommand{\esharp}{E^\sharp}

\newcommand{\edag}{E^\dagger}
\newcommand{\esd}{E_{''}}
\newcommand{\bediam}{\overline{E}^\diamond}
\newcommand{\bediamd}{\overline{E}_\diamond}

\newcommand{\bebul}{\overline{E}^\bullet}
\newcommand{\besharp}{\overline{E}^\sharp}

\newcommand{\beast}{\be^\ast}
\newcommand{\besd}{\be_{''}}
\newcommand{\hebul}{\widehat{E}^\bullet}
\newcommand{\heast}{\widehat{E}^\ast}
\newcommand{\hedag}{\widehat{E}^\dagger}
\newcommand{\hesharp}{\widehat{E}^\sharp}
\newcommand{\hediam}{\widehat{E}^\diamond}
\newcommand{\hediamd}{\widehat{E}_\diamond}
\newcommand{\hesd}{\widehat{E}_{''}}
\newcommand{\tebul}{\widetilde{E}^\bullet}
\newcommand{\tesharp}{\widetilde{E}^\sharp}
\newcommand{\tedag}{\widetilde{E}^\dagger}

\newcommand{\teast}{\widetilde{E}^\ast}
\newcommand{\tesd}{\widetilde{E}_{''}}
\newcommand{\tta}{\tilde{\tau}}
\newcommand{\hta}{\hat{\tau}}
\newcommand{\bta}{\bar{\tau}}
\newcommand{\rta}{\mathring{\tau}}
\newcommand{\tabul}{{\tau}^\bullet}
\newcommand{\tasha}{{\tau}^\sharp}
\newcommand{\ttabul}{\tilde{\tau}^\bullet}
\newcommand{\ttasha}{\tilde{\tau}^\sharp}

\newcommand{\btabul}{\bar{\tau}^\bullet}

\newcommand{\ttaast}{\tilde{\tau}^\ast}
\newcommand{\htaast}{\hat{\tau}^\ast}

\newcommand{\tadiam}{{\tau}^\diamond}
\newcommand{\btadiam}{{\bta}^\diamond}
\newcommand{\btadiamd}{{\bta}_\diamond}
\newcommand{\htadiamd}{{\hta}_\diamond}
\newcommand{\btadag}{{\bta}^\dagger}
\newcommand{\btasd}{{\bta}_{''}}

\newcommand{\htasd}{{\hta}_{''}}

\newcommand{\ttadiam}{{\tta}^\diamond}
\newcommand{\ttadag}{{\tta}^\dagger}
\newcommand{\tediam}{{\te}^\diamond}
\newcommand{\tediamd}{{\te}_\diamond}

\newcommand{\memch}{\mem{ch}}
\newcommand{\bemch}{\bem{ch}}
\newcommand{\temch}{\tem{ch}}
\newcommand{\hemch}{\hem{ch}}
\newcommand{\cross}[4]{[\,#1\,:\,#2\,|\,#3\,:\,#4\,]}
\newcommand{\ratio}[3]{[\,#1\,:\,#2\,|\,#3\,]}
\newcommand{\txtred}[1]{\textcolor{red}{#1}}

\newcommand{\mobf}{\mathfrak{M}}
\newcommand{\bmobf}{\overline{\mathfrak{M}}}
\newcommand{\tmobf}{\widetilde{\mathfrak{M}}}

\newcommand{\tritor}{\mathcal{T}^3}
\newcommand{\fourtor}{\mathcal{T}^4}

\begin{document}
\author{Jacques L. Rubin
	\bigskip \\
	Universit{\'e} de Nice--Sophia Antipolis\\
	Institut de Physique de Nice -- UMR7010-UNS-CNRS\\
	Site Sophia Antipolis\\
	1361 route des lucioles, 06560 Valbonne, France
	\bigskip \\
	\href{mailto:jacques.rubin@inphyni.cnrs.fr}{\texttt{jacques.rubin@inphyni.cnrs.fr}}}
\title{Relativistic Localizing Processes Bespeak an Inevitable Projective Geometry of Spacetime}
\date{\today}
\maketitle
\begin{abstract}
Surprisingly, the issue of events localization in spacetime is poorly understood and \textit{a fortiori} realized even in the context of Einstein's relativity. Accordingly, a comparison between observational data and  theoretical expectations might then be strongly compromised. In the present paper, we give the principles of relativistic localizing systems so as to bypass this issue. Such systems will allow to locate users endowed with receivers and, in addition, to localize any spacetime event. These localizing systems are made up of relativistic auto-locating positioning sub-systems supplemented by an extra satellite. They indicate that spacetime must be supplied everywhere with an unexpected local four dimensional projective structure besides the well-known three dimensional relativistic projective one. As a result, the spacetime manifold can be seen as a generalized Cartan space modeled on a four dimensional real projective space, \ie, a spacetime with both a local four dimensional projective structure and a compatible (pseudo-)Riemannian structure. Localization protocols are presented in details, while possible applications to astrophysics are also considered.
\end{abstract}
\section{Introduction}
The general principles of the \textit{relativistic localizing systems} have been defined in a previous paper \cite{Rubin2015} with just a few details on the projective underlying structure provided by these localizing systems. The latter are based on the so-called \textit{relativistic positioning systems} \cite{Coll:2000,bahder2001navigation,Blago:2002,Rovelli:2002gps,Coll2003jsrs,coll2006two,Coll2006aipconf,Coll2012esawrk}. 
The protocols of relativistic positioning are a priori rather simple. For instance, in a four dimensional spacetime, we can consider four emitting satellites and users with their respective (time-like) worldlines. The four emitters broadcast ``emission coordinates'' which are no more, no less than time stamps generated by on-board clocks and encoded within EM signals propagating in spacetime. Then, a so-called four-dimensional \textit{emission grid} can be constructed from this relativistic positioning system and  its system of emission coordinates. This grid can be eventually superposed to a \textit{grid of reference} supplied by a `system of reference'  (\eg, the well-known WGS84 system). And then, from this superposition, the positions of the users can be deduced in the given system of reference. More precisely, in relativistic positioning systems, the emitters not only broadcast their own time stamps, but also the time stamps they receive from the others. This process of echoes undergone by the time stamps enables users to construct the four-dimensional emission grid because they can then deduce the spacetime positions of the four emitters. And then, because the positions of the four emitters can be known also in a given system of reference the users can deduce their own positions in this system of reference from their positions in the emission grid.
\par
Here, we focus on the relativistic localizing systems which are systems incorporating relativistic positioning sub-systems. We show how \textit{causal axiomatics} \cite{EhlPiSch:72,Woodhouse73,Kingal76,Malament77} 
and particular projective structures (actually, compasses) homeomorphic to $\pru$ and $\prd$  attached all along the worldlines of the emitters of the localizing systems are sufficient to justify a four-dimensional projective structure of the spacetime; in addition to the well-known three-dimensional projective structure.
\par
Beforehand, to proceed in the difficult and delicate description of the relativistic localizing systems, we first need to define as clearly as possible the terminology and the different conventions and notations.
\section{Notations and conventions}
We consider a constellation of satellites called `\textit{emitters}' which typically broadcast numerical values (called `\textit{time stamps}') generated, for instance, by embarked onboard clocks.
\begin{enumerate}
\item The ``\textit{main}'' emitters are denoted by $\mem$, $\bem$, $\tem$ and $\hem$ with their respective worldlines $\wl$,  $\wlb$, $\wlt$ and $\wlh$. The ``\textit{ancillary}'' emitter $\sem$ and the ``\textit{user}'' $\use$ have their worldlines denoted respectively by  $\wl^\sem$ and $\wlu$.
\item The \textit{main} emitters constitute the \textit{relativistic positioning system}.
\item The \textit{ancillary} emitter $\sem$ and the main emitters constitute the  \textit{relativistic localizing systems}.
\item The event to be localized is always denoted by the small capital letter $e$.
\item \label{physid} The user $\use$ collects along its worldline all the data ---in particular, the time stamps--- from which the localization of the event $e$ is deduced. Among these data, there are those for identifying physically the event $e$ such as, for instance, its shape, its spectrum, etc., and which are obtained from apparatus making physical analyses embarked onboard each mean emitter.
\item Any explicit event will be marked by symbols like ``$\,\bullet\,$'', ``$\,\ast\,$'', ``$\tilde{\,\,\,\,\,}$'', ``$\hat{\,\,\,\,\,}$'',  etc., or also by small capital letters such as ``$\,p\,$'', ``$\,\ell\,$'', etc. Non-marked or numbered events will refer to general or generic, unspecified events. For instance, $E^\bullet$ will be a specified event while $E$ or $E_2$ will be generic, unspecified events.
\item The generic events $E$, $\be$, $\te$, $\he$, $S$ and $U$ belong respectively to the worldlines $\wl$,  $\wlb$, $\wlt$, $\wlh$, $\wl^\sem$ and $\wlu$.
\item The \textit{time stamps} will be denoted by the greek letters ``$\,\tau\,$'', ``$\,\bar{\tau}\,$'', ``$\,\tilde{\tau}\,$'', ``$\,\hat{\tau}\,$'' and ``$\,\mathring{\tau}\,$.'' The first four previous time stamps are ``generated'' and broadcast respectively by the main emitters $\mem$, $\bem$, $\tem$ and $\hem$, and the last one is ``generated'' and broadcast by the ancillary emitter $\sem$. The four main emitters not only generate their own time stamps but transmit also the time stamps they receive. These main emitters constitute the various \textit{autonomous auto-locating} relativistic positioning systems from which the relativistic localizing systems presented further are constructed.
\item Two classes of \textit{time stamps} are considered: 
\begin{itemize}
\item The time stamps which are generated and then broadcast by the emitters at given events on their worldlines.
Then, we agree to mark the corresponding time stamps like the given events. For instance, if an emitter generates and broadcasts a time stamp at the specified event $\tebul$ or at the generic event $E_1$\,, then the respective time stamps will be denoted by $\ttabul$ or $\tau_1$\,.
\smallskip 
\item The time stamps which are the \textit{emission} (or \textit{time}) coordinates of an event $K$ ---specified or not--- will be denoted by ``$\,\tau_{K}$'',
``$\,\bar{\tau}_{K}$'', ``$\,\tilde{\tau}_{K}$'', ``$\,\hat{\tau}_{K}$'' and ``$\,\mathring{\tau}_{K}$''.
\end{itemize}
\item The ancillary emitter $\sem$ generates and broadcasts its own time stamp $\rta$ and it broadcasts also its time (emission) coordinates provided by the relativistic positioning system. In other words, it is also a particular user of the relativistic positioning system like the user $\use$. Contrarily to the ancillary emitter, the user does not necessarily broadcast its emission coordinates.
\item Projective frames at events $E$ will be denoted by $\F_E$\,. There are sets of ``canonical projective points $[\ldots]_E$'' which are the following:
\begin{itemize}
\item $\F_E\equiv\Big\{[0]_E,[1]_E,[\infty]_E \Big\}$ for projective frames of the real projective line $\pru$\,, and
\item $\F_E\equiv\Big\{[0,0]_E,[1,1]_E,[0,\infty]_E,[\infty,0]_E \Big\}$ for projective frames of the 2-dimensional projective space $\prd$\,.
\end{itemize}
The subscripts will be canceled out if there are no ambiguities on the referring event.
\item The celestial circles/spheres are denoted by $\cel$, and then, $\cel_K$ is the celestial circle/hemisphere at the event~$K$. The celestial circles are invoked in the definition of the ``echoing systems'' of relativistic localizing systems in ($2+1$)-dimensional spacetimes presented in Section~\ref{21rls}\,. Considering relativistic localizing systems and  their corresponding echoing systems (Section~\ref{31rls}) in  ($3+1$)-dimensional spacetimes, then 2-dimensional projective spaces $\prd$ are also considered. But, contrarily to the relativistic localizing systems in ($2+1$)-dimensional spacetimes, the 2-dimensional real projective spaces $\prd$ cannot be immersed in  spheres $S^2$ (or $\rset^2$). Then, as well-known from the cell decomposition of  $\prd=\rset^2\cup \pru$, the Euclidean space $\rset^2$ is identified in a standard way with a hemisphere of $S^2$ while $\pru$ is identified with half of the equatorial boundary (see for instance \cite[p.10--14]{Seifert1980} for details).
\item We denote by  (see \cite[Def.~3.1, p.R16]{Garcia-Parrado2005} and \cite{Kronheimer1967})\footnote{Roughly speaking, let $x$ and $y$ be two events in spacetime. Then, 1) $x\prec y$ means that $y$ is in the future null cone of $x$ or in its interior, 2) $x\ll y$ means that $y$ is in the interior of the future null cone of $x$, and 3) $x\to y$ means that  $x$ and $y$ are joined by a null geodesic starting from $x$ to $y$. The relation of order $\to$ is \textit{reflexive} and it is also called the \textit{horismotic} relation (see for instance \cite[p.R9]{Garcia-Parrado2005}).}
\begin{itemize}
\item ``\,$\prec$\,'' the \textit{causal} order,
\item ``\,$\ll$\,'' the \textit{chronological} order, and 
\item ``\,$\to$\,'' the \textit{horismos} (or \textit{horismotic} relation/order).
\end{itemize} 
\item We call ``\textit{emission} (or \textit{positioning}) \textit{grid} $\rset^n_P$\,'' the Euclidean space $\rset^n$ of positioning, and ``\textit{localization} (or \textit{quadrometric/pentametric}) \textit{grid} $\rset^m_L$\,'' and ``\textit{anisotropic localization} (or \textit{quadrometric/pentametric}) \textit{grid} $\rset^m_{AL}$\,,'' two different Euclidean spaces $\rset^{m-1}\times\rset^*$ ascribed to two different, particular sets of time coordinates used for the localization.
\item The acronyms RPS and RLS mean respectively `\textit{Relativistic Positioning System}' and `\textit{Relativistic Localizing System}.'
\end{enumerate}
\section{RLSs in (1+1)-dimensional spacetimes }
In this (1+1)-dimensional case, there are two main emitters $\mem$ and $\bem$ constituting the RPS, and with the ancillary emitter $\sem$ they constitute the RLS. We first give the causal structures of the RPS and the associated RLS. In the Figures~\ref{caus21rps} and \ref{caus21rls} below, and also, in all other subsequent figures representing a causal structure, the arrows represent always the horismotic relation between two events.
\subsection{The causal structure of the RPS}
We have the following causal structure (see Figure~\ref{caus21rps} and Table~\ref{evbroa11}) for the auto-locating RPS from which the positioning of the user $\use$ is realized.
\par
\begin{figure}[ht]
\caption{\label{caus21rps}The causal structure of the RPS.}
$$\xy\xymatrix @!=.8em {
& U_r &
\\
E' \ar[ru] && \be' \ar[lu]
\\
&&\\
E'' \ar[uurr] && \be''  \ar[lluu] }
\endxy$$
\end{figure}
\begin{table}[ht]
\centering
\caption{\label{evbroa11}The events and their broadcast time stamps in the RPS.}
\begin{tabular}{ccc}
\hline \hline
Event & \parbox[c]{7em}{\centering broadcasts\\ time stamp(s)} & received at\\ 
\hline 
$E''$& $\tau''$ & $\be'$\\ 
$\be''$& $\bta''$ & $E'$\\
$E'$& $(\tau_{E'}=\tau',\bta'')$ & $U_r$\\
$\be'$& $(\tau'',\bta'=\bta_{\be'})$ & $U_r$\\
\hline 
\end{tabular}
\end{table} 
Then, the position of the user at the event $U_r$ in the \textit{emission grid} $\rset^2_P$ is: $U_r\equiv(\tau',\bta')$. Also, the user can know from the auto-locating process the positions of the two emitters: $E'\equiv(\tau',\bta'')$ and $\be'\equiv(\tau'',\bta')$. Moreover, ephemerides are regularly uploaded onboard the main emitters which broadcast with their time stamps these ephemerides to the users. From these data, \ie, ephemerides and positions of the main emitters, the users can deduce  their own positions with respect to a given  system of reference (\eg, the terrestrial frame of WGS84). This is the core and the important interest of the auto-locating positioning systems to immediately furnish the positions of the users with respect to a given system of reference.
\subsection{The causal structure of the RLS}
In this very specific (1+1)-dimensional case, the localized event $e$ is necessarily the intersection point of two null geodesics. The causal structure is the following (Figure~\ref{caus21rls} and Table~\ref{tabcaus21rls}):
\newpage 
\begin{figure}[ht]
\caption{\label{caus21rls}The causal structure of the RLS.}
$$\xymatrix  @!=1em {
\bu_r & \ll & U_r
\\
\be_p \ar[u] && E_p \ar[u]
\\
&e \ar[lu] \ar[ru]&
\\
\bebul \ar [ru] && \ebul \ar[lu]
}$$
\end{figure}
\begin{table}[ht]
\centering
\caption{\label{tabcaus21rls}The events and their broadcast time stamps in the RLS.}
\begin{tabular}{ccc}
\hline \hline
Event & \parbox[c]{7em}{\centering broadcasts\\ time stamp(s)} & received at\\ 
\hline 
$\bebul$& $\btabul$& $E_p$\\ 
$\ebul$& $\tabul$ & $\be_p$\\ 
$\be_p$& $(\tabul,\bta_p=\bta_{\be_p})$ & $\bu_r$\\
$E_p$& $(\tau_{E_p}=\tau_p,\btabul)$& $U_r$\\
\hline 
\end{tabular}
\end{table}
Then, the protocol of localization  gives the following time coordinates for $e\equiv(\tau_e,\bta_e)$ in the \textit{localization grid} $\rset^2_L$: $\tau_e=\tau_{\ebul}=\tabul$ and $\bta_e=\bta_{\bebul}=\btabul$.
\begin{remark}\label{match}
It matters to notice that the two events of reception $U_r$ and $\bu_r$ are matched by the user on the basis of a crucial identification of the physical data transmitted by the two main emitters (see convention \ref{physid}) and which allow to explicitly identify the physical occurrence of an event $e$. And then, the whole of different time stamps collected at these two events can be therefore considered by the user as those needed to make the localization of $e$.
\end{remark}
\subsection{Consistency between the positioning and localizing protocols -- Identification}
\begin{definition}\label{consist}
\textbf{Consistency} -- We say that the localizing and the positioning protocols or systems are ``\emph{con\-sis\-tent}'' if and only if the time coordinates $(\tau_K,\bta_K,\ldots)$ ascribed to each event $K$ belonging to an emitter's worldline and provided by the localization (\emph{resp.} positioning) system are the same as those provided by the positioning (\emph{resp.} localization) system.
\end{definition}
\begin{remark}
\label{equivlocpos}
In this (1+1)-dimensional case, when we \emph{identify} the time stamps $\tau_e$ and $\bta_e$ with, respectively, $\tau_{\ebul}$ and $\bta_{\bebul}$, then the localization is equivalent to the positioning. This leads to the  general Definition~\ref{iden} below.
\end{remark}
\begin{remark}
\label{noconsist}
The consistency between the localizing and the positioning protocols is not an absolute necessity. We can obtain different time coordinates for the same event $K$ belonging to an emitter's worldline from the positioning system or the localizing one if we change the time stamps ascriptions in the protocols of localization presented further. Then, we can choose arbitrarily the emission grid or the grid of localization to position the event $K$, and then, we can refer to the preferred grid for the time coordinates ascribed to any other event, positioned or localized. In other words, because the  systems of localization include implicitly by construction derived positioning systems, the latter can differ from the initial ones. In this case, the consistency is not satisfied but we can still refer the time coordinates of any event with respect to the localization grids rather than to the emission grids. The only advantage of the consistency is that once the events are localized, then they time coordinates can be ascribed indifferently to any of the two grids.
\end{remark}
\begin{definition}\label{iden}
Let a localized event $e$ and an event $K$ on the worldlines of a \emph{main} emitter or of the ancillary emitter be such that $e\to K$ or $K\to e$ or $e=K$. Then,  we call `\emph{identification}' in the emission (position) grid the ascription of an emission coordinate of $e$ to an emission coordinate of $K$.
\end{definition}
\section{RLSs in (2+1)-dimensional spacetimes\label{21rls}}
In this case, there are three main emitters $\mem$, $\bem$ and $\tem$ constituting the auto-locating RPS and, again, an ancillary emitter $\sem$ with which they constitute the RLS.
\subsection{The causal structure of the RPS}
This causal structure is described in Figure~\ref{cuasstr21} and Table~\ref{tab21}.
\par 
\begin{figure}[ht]
$$\xy\xymatrix @!=1em {
&& U_r &&\\
E' \ar[rru] && \be' \ar[u]&& \te' \ar[llu]\\
&&&&\\
E'' \ar[rruu]&& {}_{\displaystyle\be''}\ar[rruu]  && \ar[lluu]\te''\\
%&&&&\\
**[r]E'''\ar[rrrruuu] && {}_{\displaystyle\be'''}\ar[lluuu]&& \ar[lllluuu]\te'''\\ }
\endxy$$
\caption{\label{cuasstr21}The causal structure of the RPS in a (2+1)-dimensional spacetime.}
\end{figure}
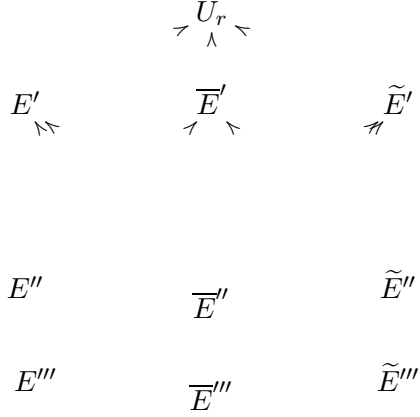
%\newpage
%
\begin{table}[ht]
\centering
\caption{\label{tab21}The events and their broadcast time stamps.}
\begin{tabular}{ccc}
\hline \hline
Event & \parbox[c]{7em}{\centering broadcasts\\ time stamp(s)} & received at\\ 
\hline
$E'''$& $\tau'''$ & $\te'$\\ 
$\be'''$& $\bta'''$ & $E'$\\ 
$\te'''$& $\tta'''$ & $E'$\\
$E''$& $\tau''$ & $\be'$\\ 
$\be''$& $\bta''$ & $\te'$\\ 
$\te''$& $\tta''$ & $\be'$\\
$E'$& $(\tau_{E'}=\tau',\bta''',\tta''')$ & $U_r$\\
$\be'$& $(\tau'',\bta'=\bta_{\be'},\tta'')$ & $U_r$\\
$\te'$& $(\tau''',\bta'',\tta_{\te'}=\tta')$ & $U_r$\\
\hline 
\end{tabular}
\end{table}
Then, the position in the emission grid $\rset_P^3$ of the user at $U_r$ is: $(\tau',\bta',\tta')$, and those of $E'$, $\be'$ and $\te'$ are respectively: $(\tau',\bta''',\tta''')$, $(\tau'',\bta',\tta'')$ and
$(\tau''',\bta'',\tta')$.
\begin{remark}\label{rk4}
It matters to notice that in auto-locating RPSs the time stamp broadcast by each main emitter is also one of its emission coordinates, \eg, $\tau_{E'}=\tau'$ for $\mem$ at $E'$ in Table~\ref{evbroa11} and $\bta'=\bta_{\be'}$ for $\bem$ at $\be'$ in Table~\ref{tab21}. This property is common to any RPS whatever is the spacetime dimension. 
\end{remark}
\subsection{The description of the RLS and its causal structures}
The determination of the first emission coordinate  $\tau_e$ for the event $e$ to be localized is obtained from a first system of light ``echoes'' associated with the privileged emitter $\mem$. And then, this system is linked to one event of reception $U_r\in\wlu$ where all the time stamps are collected by the user. We denote by $\memch$ this system of light ``echoes'' on the worldline of the given, privileged emitter $\mem$.
\par
Also, one of the key ingredient in the echoing process presented below is the way any event $K$ in the past null cone of $E_p$ is associated with a ``bright'' point on the celestial circle $\cel_{E_p}$ (see Figure~\ref{pstcne21}). Because  $K\to E_p$\,, we can only consider  null ``directions'' $k_{E_p}$ at the origin $E_p$ and tangent at $E_p$ to the null geodesic joining $K$ to $E_p$. The abstract space whose element are these past null directions we call $\nm$\,. This space can be represented by the intersection $\cel_{E_p}$ of the past null cone with a spacelike surface passing through an event $N_p\in\wl$ in the past vicinity of $E_p$, \ie, $N_p\ll E_p$\,. Then, the exterior of this celestial circle represents spacelike directions.
\par
In physical terms, the significance of $\cel_{E_p}$  is the following. Light rays reaching the event $E_p$ and detected by the ``eye'' of the satellite correspond to null lines through $E_p$ whose past directions constitutes the field of vision of the ``observing'' satellite. This is $\nm$ and it is represented by the circle $\cel_{E_p}$\,; the latter to be an accurate geometrical representation of what the satellite actually `` sees.'' For, the satellite can be considered as permanently situated at the center of a unit circle (his circle of vision) onto which the satellite maps all it detects at any instant. Then, the mapping  of the past null directions at $E_p$ to the points of $\cel_{E_p}$ we can call the \textit{sky mapping}. 
Additionally,  because 1) the circle $S^1$ is homeomorphic to the real projective line $\pru$\,, and 2) we need angle measurements to frame the points of $\cel_{E_p}$ associated with any event $K$ in the past null cone of $E_p$ to be furthermore localized, then a particular production process of projective frame for $\cel_{E_p}$ must be devised and incorporated in the echoing system definition now given below.
\begin{definition}
\label{echsys}
\textbf{The echoing system $\memch$} --
The echoing system $\mem{ch}$ associated with the \emph{privileged} emitter $\mem$  is based on the following features (see Figure~\ref{causechoes} and Table~\ref{tabcausechoes}): 
\begin{itemize}
\item one \emph{primary} event $E_p$ with its celestial circle $\cel_{E_p}$,
\item three \emph{secondary} events $\bebul$, $\tebul$ and $\sembul$, associated respectively with the canonical projective points $[0]_{E_p}$, $[\infty]_{E_p}$ and $[1]_{E_p}$ of the projective frame $\F_{E_p}$ defined on $\cel_{E_p}$,
\item two \emph{ternary} events: $\ediam$ and $E''$, 
\item a compass on $\cel_{E_p}$ with a moving origin \emph{anchored} on the projective point $[1]_{E_p}$ of $\cel_{E_p}$ associated with $\sembul$,
\item an event of reception $U_r\in\wlu$ at which all the data are collected and sent by the emitter $\mem$.
\end{itemize}
\end{definition}
The determination of the second (\textit{resp}. third) emission coordinate  $\bta_e$ (\textit{resp}. $\tta_e$) for the event $e$ to be localized is obtained from a second (\textit{resp}. third) system of ``echoes'' associated with the privileged emitter $\bem$ (\textit{resp}. $\tem$). It is also linked to one event of reception $\bu_\ell$ (\textit{resp}. $\tu_\ell$) where all the time stamps are collected. We denote by $\bemch$ (\textit{resp}. $\temch$) this second (\textit{resp}. third) system of ``echoes'' on the worldline of the privileged emitter $\bemch$ (\textit{resp}. $\temch$).
\par
Then, we have:
\begin{definition}
\label{tbechsys}
\textbf{The echoing systems $\bemch$ and $\temch$} --
The definitions of the echoing systems $\bemch$ and $\temch$ are obtained when making the following substitutions of events and marks in the definition of $\memch$:
\begin{itemize}
\item For $\bemch$\,: $(U,E,\be,\te)\longrightarrow(\bu,\be,\te,E)$ and \,$\bullet\longrightarrow\ast$\,, and
\item for $\temch$\,: $(U,E,\be,\te)\longrightarrow(\tu,\te,E,\be)$ and \,$\bullet\longrightarrow\,'\,$\,. 
\end{itemize}
\end{definition}
\newpage 
Then, we have the following causal structure of the echoing system $\memch$ (Figure~\ref{causechoes} and Table~\ref{tabcausechoes}); the other two causal structures for $\bemch$ and $\temch$ (Figures~\ref{causbechoes} and \ref{caustechoes}) are deduced from the causal structure of $\memch$ by making the substitutions indicated in Definition~\ref{tbechsys}. We indicate also the three structures with the event $e$ (Figure~\ref{threeS}).
\begin{figure}[ht]
\centering
$$\xymatrix  @!=1em {
&& U_r &
\\
*!R{primary}&& E_p \ar[u]&
\\
*!R{secondary}&\underset{\txtred{([0])}}{\bebul} \ar[ru] &\underset{\txtred{([\infty])}}{\tebul} \ar[u] & \underset{\txtred{([1])}}{\sembul} \ar[lu]
\\
*!R{ternary}&\ediam \ar[u] &E'' \ar[u] & 
\\
&&& {e} \ar[luuu]
}$$
\textit{$\ediam$ and $E''$ chronologically ordered.}
\caption{\label{causechoes}The causal structure of $\memch$ with the event $e$\,.}
\end{figure}
%\par
%
\begin{table}[ht]
\centering
\caption{\label{tabcausechoes}The events and their broadcast time stamps in the $\memch$ system.}
\begin{tabular}{ccc}
\hline \hline
Event & \parbox[c]{7em}{\centering broadcasts\\ time stamp} & received at\\ 
\hline 
$\ediam$& $\tadiam$& $\bebul$\\ 
$E''$& $\tau''$ & $\tebul$\\
$\sembul$& $\tau_{\sembul}$ & $E_p$\\ 
$\bebul$& $\tadiam$ & $E_p$\\
$\tebul$& $\tau''$ & $E_p$\\
\hline 
\end{tabular}
\end{table}
\begin{figure}[!h]
\centering
$$\xymatrix  @!=1.2em {
&& \bu_r &
\\
*!R{primary}&& \be_p \ar[u]&
\\
*!R{secondary}&\underset{\txtred{([0])}}{\teast} \ar[ru] &\underset{\txtred{([\infty])}}{\east} \ar[u] & \underset{\txtred{([1])}}{\semast} \ar[lu]
\\
*!R{ternary}&\bediam\ar[u] &\be'' \ar[u] & 
\\
&&& {e} \ar[luuu]
}$$
\textit{$\bediam$ and $\be''$ chronologically ordered.}
\bigskip
\caption{\label{causbechoes}The causal structure of $\bemch$ with the event $e$\,.}
\end{figure}
%\newpage
%
\begin{figure}[ht]
\centering
$$\xymatrix  @!=1em {
&& \tu_r &
\\
*!R{primary}&& \te_p \ar[u]&
\\
*!R{secondary}&\underset{\txtred{([0])}}{E'} \ar[ru] &\underset{\txtred{([\infty])}}{\be'} \ar[u] & \underset{\txtred{([1])}}{\sem'} \ar[lu]
\\
*!R{ternary}&\tediam \ar[u] &\te'' \ar[u] & 
\\
&&& {e} \ar[luuu]
}$$
\textit{$\tediam$ and $\te''$ chronologically ordered.}
\bigskip
\caption{\label{caustechoes}The causal structure of $\temch$ with the event $e$\,.}
\end{figure}
\begin{figure}[ht]
\centering
$$\xymatrix  @!=1em {
&U_r&\bu_p&\tu_p&&&
\\
&E_p\ar[u]&\be_p\ar[u]&\te_p\ar[u]&&&
\\
&&&&&&
\\
e\ar[uur]\ar[uurr]\ar[uurrr]&&\sembul\ar[uul]&\ll&\semast\ar[uull]&\ll& \sem'\ar[uulll]
}$$
\caption{\label{threeS}The causal structure for the three echoing systems $\memch$, $\bemch$ and $\temch$ with the event~$e$. The chronological order between $\sembul$, $\semast$ and $S'$ belonging to $\wl^\sem$ can be different.}
\end{figure}
%\par
%
\begin{remark}
Again (Remark~\ref{match}), it matters to notice that the three events of reception $U_r$, $\bu_r$ and $\tu_r$ (Figure~\ref{threeS}) are matched by the user on the basis of an identification of the physical data for $e$ transmitted by the main emitters (see convention \ref{physid}).
\end{remark}
\subsection{The projective frames and the time stamps correspondences}
The realization of the RLS is based on a sort of spacetime parallax, \ie, a passage from angles ``$\alpha$'' measured on celestial circles to spatio-temporal distances. And thus, because spatio-temporal distances are evaluated from time stamps ``$\tau$'' in the present context, we need to make the translation of angles into time stamps. This involves onboard compasses embarked on each main emitter to find somehow the bearings. Then, this translation is neither more nor less than a change of projective frames.
\par
To make this change of projective frames effective, we need to define the projective frames on the celestial circles attached to each main emitter. This can be done ascribing to specific ``bright points'' detected on the celestial circles both angles and time-stamps.  This ascription is then naturally achieved if these bright points are the main emitters themselves since they broadcast the time-stamps. But, if we have three emitters for the RPS, then only two bright points can be detected on each celestial circle attached to each main emitter. And, we need three bright points to have a projective frame on the celestial circle homeomorphic to $\pru$; thus the need for the ancillary emitter $\sem$.
%\newpage
% 
The change of projective frames is described in Table~\ref{chprjfr21} and Figures~\ref{pstcne21} and \ref{prjline21}. For instance, the main emitter $\bem$  broadcasts the time stamp $\tadiam$ at the secondary event $\bebul$, and the former is then received by the emitter $\mem$ at the primary event $E_p$. Also, if $\bem$ is always associated by convention with the canonical projective point $[0]_{E_p}$ on the celestial circle of $\mem$, then we deduce that $\tadiam$ corresponds by a projective transformation to $0$. And then, we proceed in the same way with the other two canonical projective points. 
\par
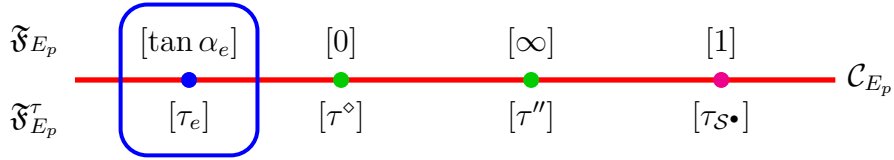
\begin{table}[ht]
\centering
\caption{\label{chprjfr21}The change of projective frame and the corresponding events.}
\begin{tabular}{ccccc}
\hline \hline
Event &\mbox{\qquad}& $\F_{E_p}$ &\mbox{\qquad}& $\F^\tau_{E_p}$\\ 
\hline 
$e$&& $[\tan\alpha_e]$&& $[\tau_e]$\\ 
$\bebul$&& $[0]$ && $[\tadiam]$\\
$\tebul$&& $[\infty]$ && $[\tau'']$\\ 
$\sembul$&& $[1]$ && $[\tau_{\sembul}]$\\
\hline 
\end{tabular}
\end{table}
\begin{figure}[ht]
\centering
\caption{\label{pstcne21}The past null cone and the celestial circle $\cel_{E_p}$.}
\begin{tikzpicture}[scale=.85]
\large
\fill[
left color=yellow!90!white,
right color=yellow!70!white, 
shading=axis,
opacity=0.9
](4,0) -- (0,3) -- (-4,0) arc (180:360:4cm and 1cm);
%%%%
\draw[color=red,ultra thick](2.7,0.97)node[anchor=south west,color=black]{$\cel_{E_p}$}  arc (0:180:2.7cm and -0.54cm) ;
%%%%%
\draw[opacity=1,thick,color=orange!1](4,0) -- (0,3) node[anchor=south,color=black]{$E_p$} -- (-4,0) arc (180:360:4cm and 1cm);
\draw[-Stealth,line width=1pt,color=magenta](4,-1.5) node[anchor=north  ,color=black]{$\sembul$} -- (0,3) ;
\fill[color=magenta] (2.1,0.63)node[right=5pt,above=1pt,color=black]{$\boldsymbol{[1]}$} circle [radius=3pt];
\draw[-Stealth,line width=1pt,color=green!80!black](-2.5,-2) node[anchor=north,color=black]{$\bebul$} -- (0,3) ;
\fill[color=green!80!black] (-1.25,0.5)node[left=5pt,above=1pt,color=black]{$\boldsymbol{[0]}$} circle [radius=3pt];
\draw[-Stealth,line width=1pt,color=green!80!black](0.5,-2.1)  node[anchor=north,color=black]{$\tebul$}-- (0,3) ;
\fill[color=green!80!black] (.25,0.45)node[left=12pt,below=2pt,color=black]{$\boldsymbol{[\infty]}$} circle [radius=3pt];
\draw[-Stealth,line width=1pt,color=blue] (-4.5,-1.5) 
node[anchor=north east,color=black]{$e$} -- (0,3) ;
\fill[color=blue] (-2.3,0.7) node[left=25pt,below=2pt,color=black]{$\boldsymbol{[\tan\alpha_e]}$}  circle [radius=3pt];
\end{tikzpicture}
\end{figure}
\begin{figure}[ht]
\centering
\caption{\label{prjline21}The projective line associated with the celestial circle $\cel_{E_p}$.}
\begin{tikzpicture}[scale=1]
\draw[line width=2pt,color=red](-5,0)
node[left=15pt,below=4pt,color=black]{\large$\F^{\tau}_{E_p}$}
node[left=15pt,above=4pt,color=black]{\large$\F_{E_p}$}
-- (5,0) node[color=black,right]{\large$\cel_{E_p}$} ;
\fill[color=blue] (-3.5,0) node[above=4pt,color=black]{\large$[\tan\alpha_e]$}
node[below=4pt,color=black]{\large$[\tau_e]$}
circle [radius=3pt];
\draw[rounded corners=10pt,line width=1.5pt,color=blue] (-4.4,-1) rectangle (-2.6,1);
\fill[color=green!80!black] (-1.5,0) node[above=4pt,color=black]{\large$[0]$}
node[below=4pt,color=black]{\large$[\tadiam]$}
circle [radius=3pt];
\fill[color=green!80!black] (1,0) node[above=4pt,color=black]{\large$[\infty]$}
node[below=4pt,color=black]{\large$[\tau'']$}
circle [radius=3pt];
\fill[color=magenta] (3.5,0) 
node[above=4pt,color=black]{\large$[1]$}
node[below=4pt,color=black]{\large$[\tau_{\sembul}]$}
circle [radius=3pt];
\end{tikzpicture}
\end{figure}
As a result, the relations between the angles and the time-stamps are the following:
\begin{subequations}
\label{crossr}
\begin{align}
&\tan\alpha_e
=\cross{\tadiam}{\tau''}{\tau_e}{\tau_{\sembul}}
=\frac{\ratio{\tadiam}{\tau''}{\tau_e}}{
\ratio{\tadiam}{\tau''}{\tau_{\sembul}}}
\equiv\mobf(\tau_e)\,,
\\
&\tan\bar{\alpha}_e
=\cross{\btadiam}{\bta''}{\bta_e}{\bta_{\semast}}
=\frac{\ratio{\btadiam}{\bta''}{\bta_e}}{
\ratio{\btadiam}{\bta''}{\bta_{\semast}}}
\equiv\bmobf(\bta_e)\,,
\\
&\tan\tilde{\alpha}_e
=\cross{\ttadiam}{\tta''}{\tta_e}{\tta_{\sem'}}
=\frac{\ratio{\ttadiam}{\tta''}{\tta_e}}{
\ratio{\ttadiam}{\tta''}{\tta_{\sem'}}}
\equiv\tmobf(\tta_e)\,,
\end{align}
\end{subequations}
where $\cross{a}{b}{c}{d}$ is the cross-ratio of the four projective points
$a$, $b$, $c$ and $d$:
\begin{equation}
\label{crossratio}
\cross{a}{b}{c}{d}=\frac{\ratio{a}{b}{c}}{\ratio{a}{b}{d}}
\qquad
\text{where}
\qquad
\ratio{a}{b}{c}=\cross{a}{b}{c}{\infty}=\left(\frac{a-c}{b-c}\right).
\end{equation}
Conversely, the time coordinates for the event $e$ are then obtained from the angles measurements and the following formulas:
\begin{subequations}
\label{mobtau}
\begin{align}
&\tau_e=\left(
\frac
{\tadiam-\tau''\ratio{\tadiam}{\tau''}{\tau_{\sembul}}\tan\alpha_e}
{1-\ratio{\tadiam}{\tau''}{\tau_{\sembul}}\tan\alpha_e}
\right),\\
&\bar{\tau}_e=\left(
\frac
{\btadiam-\bta''\ratio{\btadiam}{\bta''}{\bta_{\semast}}\tan\bar{\alpha}_e}
{1-\ratio{\btadiam}{\bta''}{\bta_{\semast}}\tan\bar{\alpha}_e}
\right),\\
&\tilde{\tau}_e=\left(
\frac
{\ttadiam-\tta''\ratio{\ttadiam}{\tta''}{\tta_{\sem'}}\tan\tilde{\alpha}_e}
{1-\ratio{\ttadiam}{\tta''}{\tta_{\sem'}}\tan\tilde{\alpha}_e}
\right).
\end{align}
\end{subequations}
And thus, the event $e$ is localized in the localization grid $\rset^3_L$.
Then, we deduce the following lemma:
\begin{lemma}
The map 
\begin{equation}\label{autol3}
\mathbb{S}:(\tan\alpha_e,\tan\bar{\alpha}_e,\tan\tilde{\alpha}_e)\in\tritor
\longrightarrow(\tau_e,\bta_e,\tta_e)\in\tritor
\end{equation}
where the $\tritor\equiv(\pru)^3$ is an automorphism.
\end{lemma}
\begin{proof}
Obvious from the relations \eqref{crossr}, because $\mobf$, $\bmobf$ and $\tmobf$ are bijective M{\"o}bius transformations. 
\end{proof}
\subsection{The Consistency  between the positioning and localization protocols}
\begin{theorem}\label{consist21}
The localization and the positioning protocols or systems in a (2+1)-dimensional spacetime are consistent.
\end{theorem}
\begin{proof}\label{proofconsist21}
The consistency must be satisfied if $e$ is an element of the emitters' worldlines. Indeed, the localization protocol is consistent with the positioning protocol if the set of events on the emitters' worldlines from which the localization of any event $e$ is possible are themselves localizable.
\begin{itemize}
\item 
\textbf{Case 1: $e\in\wl^\sem$.}\par
We consider two cases: $e\ll\sembul$ and $\sembul\ll e$. The other cases with $\semast$ or $S'$ instead of $\sembul$ give the same results. Now, we start with the assumption  $e\ll\sembul$ from which we deduce the following causal structure:
$$\xymatrix  @!=1em {
&&&E_p\\
\sembul \ar[urrr]&\ll& e \ar[ur] &
}$$
In particular, from $\sembul\ll e\to E_p$\,, we find\footnote{
\label{conVII}
In \cite{Kronheimer1967}:
\begin{itemize}
\item Condition (V): $x\ll y \Longrightarrow x\prec y$.
\item Condition (VII): $x\to y\Longleftrightarrow x\prec y$ and $x\not\ll y$.
\end{itemize}}
that $\sembul\prec e\prec E_p$, and then, with $\sembul\to E_p$, we obtain\footnote{\textbf{Lemma 1-1}\cite{Kronheimer1967}: Let $x$, $y$ and $z$ be points in a causal space. If $x\prec y\prec z$ and $x\to z$ then $x\to y\to z$\,.} $\sembul\to e\to E_p$\,. But then,${}^{\ref{conVII}}$ we have  $\sembul\not\ll e$. With the assumption $e\ll\sembul$, we deduce also $e\not\ll\sembul$,  and therefore $\sembul=e$\,. Hence, we consider that $\sembul=e$ and with the other two sets of events $\semast$ with $\be_p$ or $S'$ and $\te_p$, we deduce finally that $e=\sembul=\semast=S'\equiv S$\,. Therefore, we conclude that the time coordinates of $e$ provided by the positioning system are $\tau_e=\tau_S$, $\bta_e=\bta_S$ and $\tta_e=\tta_S$\,.
\par
Besides, from the projective frames, we have also:
\begin{subequations}\label{crossrbis}
\begin{align}
&\tan\alpha_e
=\cross{\tadiam}{\tau''}{\tau_e}{\tau_{\sembul}}
=1\,,
\\
&\tan\bar{\alpha}_e
=\cross{\btadiam}{\bta''}{\bta_e}{\bta_{\semast}}
=1\,,
\\
&\tan\tilde{\alpha}_e
=\cross{\ttadiam}{\tta''}{\tta_e}{\tta_{\sem'}}
=1\,.
\end{align}
\end{subequations}
And therefore, we obtain:
\begin{equation}\label{egaltau}
\tau_e=\tau_{\sembul}\,,
\qquad
\bta_e=\bta_{\semast}\,,
\qquad
\tta_e=\tta_{\sem'}\,,
\end{equation}
which are the coordinates of $S$.
\par
In conclusion, the localization protocol is consistent with the positioning one.
\item
\textbf{Case 2: $e$ is a \textit{primary} event: $e=E_p\,(primary)\in\wl$.}
\par
In this case, we obtain the following causal structure (Figure~\ref{threeP}):
\par
\begin{figure}[ht]
$$\xymatrix  @!=3em {
&\bu_r&&U_r&&\tu_r&\\
&\be_P\ar[u]&&&&\te_P\ar[u]&\\
\semast\txtred{([1])}\ar[ur]&\teast\txtred{([0])}\ar[u]&&e=E_P=\east\txtred{([\infty])}=E'\txtred{([0])}\ar[ull]\ar[uu]\ar[urr]&&\be'\txtred{([\infty])}\ar[u]&\sem'\txtred{([1])}\ar[ul]\\
&\bediam\ar[u]&\bebul\txtred{([0])}=\be''\ar[ur]&\sembul\txtred{([1])}\ar[u]&\tebul\txtred{([\infty])}=\tediam\ar[ul]&\te''\ar[u]&\\
&&\ediam\ar[u]&&E''\ar[u]&&
}$$
\caption{\label{threeP}The causal structure for the three echoing protocols $\memch$, $\bemch$ and $\temch$ whenever $e=E_P$.}
\end{figure}
Then, from the three echoing causal structures $\memch$, $\bemch$ and $\temch$, we have $e=E_p=\east=E'$ where $\east$ is associated with the projective point $[\infty]_{\be_p}$ and $E'$ is associated with the projective point $[0]_{\te_p}$. Consequently, we have $\tan\bar{\alpha}_e=\infty$ and $\tan\tilde{\alpha}_e=0$. Also, we have $\bebul=\be''$ and $\tebul=\tediam$ from which we deduce from the positioning system that their time coordinates are equal, \ie, we have (one of the emission coordinates is equal to the broadcast one in the positioning protocol; see Remark~\ref{rk4})
\begin{equation}\label{eepos}
\bta_{\bebul}=\btabul=\bta''=\bta_{\be''}\,,\qquad
\tta_{\tebul}=\ttabul=\ttadiam=\tta_{\tediam}\,. 
\end{equation}
 Besides, from the localization protocol, we have
\begin{subequations}\label{crossrter}
\begin{align}
&\tan\alpha_e
=\cross{\tadiam}{\tau''}{\tau_e}{\tau_{\sembul}}
=\,?\quad\text{(not defined)}\,,
\\
&\tan\bar{\alpha}_e
=\cross{\btadiam}{\bta''}{\bta_e}{\bta_{\semast}}
=\infty\,,
\\
&\tan\tilde{\alpha}_e
=\cross{\ttadiam}{\tta''}{\tta_e}{\tta_{\sem'}}
=0\,.
\end{align}
\end{subequations}
Hence, we deduce:
\begin{equation}\label{casa}
\tau_e=\,?\,,
\qquad
\bta_e=\bta''\,,
\qquad
\tta_e=\ttadiam\,.
\end{equation}
And from the positioning protocol, because $E_p$ is a positioned point with emission coordinates $(\tau_{E_p},\btabul,\ttabul)$, we have also
\begin{equation}\label{lcons}
\bta_e=\bta_{E_p}=\btabul\,,
\qquad
\tta_e=\tta_{E_p}=\ttabul\,,
\end{equation}
and therefore with \eqref{eepos}, we deduce the consistency for two time stamps. Actually, $\tau_e$ is not obtain by localization but by \textit{identification} (Definition~\ref{iden}). Indeed, we know that $e$ is an element of $\wl$ and that $\tau_e=\tau_{E_p}$ is broadcast by the \textit{identified} main emitter $\mem$. This determination of $\tau_e$ is then similar to the emission coordinate ascription presented in the (1+1)-dimensional case for which localization is equivalent to positioning (Remark~\ref{equivlocpos}); hence the consistency.
\item \textbf{Case 3: $e$ is a \textit{secondary} event: $e=\bebul\,(secondary)\in\wlb$ or $e=\be'\,(secondary)\in\wlb$.}
Then, the causal structure is the following whenever $e=\bebul$ (Figure~\ref{threePS}):
\begin{figure}[ht]
$$\xymatrix  @!=5.5ex {
&\tu_r&&&U_r&&&\bu_r&\\
&\te_P\ar[u]&&&E_P\ar[u]&&&\be_P\ar[u]&\\
\sem'\txtred{([1])}\ar[ur]&E'\txtred{([0])}\ar[u]&\be'\txtred{([\infty])}\ar[ul]&\tebul\txtred{([\infty])}\ar[ur]& {e=\bebul\txtred{([0])}}\ar[u]\ar[ulll]\ar[urrr]&\sembul\txtred{([1])}\ar[ul]&\teast\txtred{([0])}\ar[ur]&\east\txtred{([\infty])}\ar[u]&\semast\txtred{([1])}\ar[ul]\\
&\tediam\ar[u]&\te''\ar[u]&E''\ar[u]&\ediam\ar[u]&&\bediam\ar[u]&\be''\ar[u]&
}$$
\caption{\label{threePS}The causal structure for the three echoing protocols $\memch$, $\bemch$ and $\temch$ whenever $e=\bebul$.}
\end{figure}
\begin{itemize}
\item $e=\bebul\,(secondary)\in\wlb$ -- Then, the localization protocol at $E_p$ gives the formula: $\tan\alpha_e=\cross{\tadiam}{\tau''}{\tau_e}{\tau_{\sembul}}=0$ because $\bebul$ is associated with the projective point $[0]_{E_p}$. Therefore, we have $\tau_e=\tadiam$. But, from the positioning protocol, the emission coordinate $\tau_{\bebul}$ of $\bebul$ relative to the main emitter $\mem$ is $\tadiam$ broadcast at the ternary event $\ediam$. Hence, $\tau_{\bebul}=\tadiam$ and we deduce the consistency of the localization protocol with the positioning protocol for one emission coordinate. 
\item $e=\be'\,(secondary)\in\wlb$ -- The reasoning is similar to the previous one. Then, we deduce the consistency for $\tta_e=\tta_{\be'}=\tta''$ because 1) $\te_p$ is the primary event for $\be'$ and $\te''$ is the ternary event for $\be'$, and 2) $\cross{\ttadiam}{\tta''}{\tta_e}{\tta_{\sem'}}=\infty$ which involves $\tta_e=\tta''$.
\end{itemize} 
Now, we consider two distinct causal structures of localization $a$ and $b$ such that $e=\be^{\bullet{}a}=\be'^{b}$ from which we deduce the consistency for  $\tau_e$ and $\tta_e$\,. Furthermore, as in the case 2, we deduce $\bta_e$ by \textit{identification} (Definition~\ref{iden}) and we obtain $\bta_e=\bta_{\be^{\bullet{}a}}=\bta_{\be'^{b}}$\,; hence the consistency.
\item \textbf{Case 4: $e$ is a \textit{ternary} event.}
\par
For instance, we can set $e=\ediam$. But then, we have also $e=\ediam\to E_P$ on $\wl^\mem$ which is impossible since we have only the chronological order on the emitters' worldlines.
\end{itemize}
\end{proof}
\begin{remark}\label{selfloc}
From this theorem, we can then notice that RLSs are based on \emph{auto-localization} protocols similarly to RPSs which are based on \emph{auto-location} protocols. As a result, RLSs and RPSs are independent of any \emph{system of reference}.
\end{remark}
\subsection{The local projective structure}
\begin{definition}\label{grids21}
We call
\begin{itemize}
\item \emph{Emission grid}, the Euclidean space $\rset^3_P\equiv\rset^3$ of the positioned events $\ep=(\tau_e,\bta_e,\tta_e)$,
\item \emph{Localization (\emph{or} quadrometric) grid}, the Euclidean space $\rset^4_L\equiv\rset^3\times\rset^\ast$ of the localized events $\el=(\tau_e,\bta_e,\tta_e,\rta_e)$ where $\rta_e$ is provided by the ancillary emitter $\sem$ by \emph{identification} from the horismotic relation $S^L\to\ep$ $(S^L\in\wl^\sem)$ or the `\emph{message function}' \cite{Woodhouse73} $f_\sem^-:\rset^3_P\longrightarrow\wl^\sem$, \ie, the time stamp $\rta^L$ broadcast by $\sem$ at $S^L$ is such that $\rta^L\equiv\rta_e$, and
\item \emph{Anisotropic localization (\emph{or} quadrometric) grid}, the Euclidean space $\rset^4_{AL}\equiv\rset^3\times\rset^\ast$ of  events $\eal=(\rta_e\tau_e,\rta_e\bta_e,\rta_e\tta_e,\rta_e)$.
\end{itemize}
\end{definition}
\begin{definition}
We denote by $I:\rset_L^4\longrightarrow
\rset_{AL}^4$\,, the bijective map such that $I(\el)=\eal$. And we denote by $\pi:\rset_{AL}^4\longrightarrow
\rset_{P}^3$\,, the submersion such that $\pi(\eal)=\ep$.
\end{definition}
\begin{remark}
In these definitions, the time coordinate $\rta_e$ must be non-vanishing. If this condition is not satisfied we can, nevertheless, always consider that the ancillary emitter $\sem$ \emph{generates} a time number $\mathring\rho$ and \emph{broadcasts} $e^{\mathring\rho}\equiv\rta$. This can be realized from a real-time computer with $\mathring\rho$ as the generated input and $e^{\mathring\rho}$ as the broadcast output. Obviously, we can assume the same for the main emitters.
\end{remark}
Let $g$ be an element of $GL(4,\rset)$ such that $g\,.\,\eal=\eal'$\,. And thus, $GL(4,\rset)$ acts linearly on $\rset_{AL}^4$. Then, the action of $GL(4,\rset)$ on $\rset_{L}^4$ and $\rset_P^3$ is non-linear, locally transitive and it defines homographies (\ie, conformal transformations):
\begin{subequations}\label{hom}
\begin{align}
&\ep'=\left(\frac{A\,.\,\ep+b}{c\,.\,\ep+\mu}\right),
\qquad
g\equiv\begin{pmatrix}
A& b \\ 
 {}^ tc&\mu
\end{pmatrix},\label{epa}
\\
&\rta'_e=\rta_e(c\,.\,\ep+\mu)\,,
\end{align}
\end{subequations}
where $\mu\in\rset$\,, $(b,\,c)\in(\rset^3)^2$ and $A\in M_{3\times3}(\rset)$\,. Let us notice that $\rta_e$ does not intervene in \eqref{epa}. Moreover, we deduce that $PGL(4,\rset)$ acts locally transitively on $\rset_P^3$\,. Therefore, we obtain:
\begin{theorem}\label{pgl4}
The (2+1)-dimensional spacetime manifold has a \emph{local} 3-dimensional projective structure inherited from its causal structure.
\end{theorem}
\begin{proof}
Let $x$, $t$ and $s_i$ ($i=1,2,3$) in $\rset^4_{AL}$ be such that
\begin{align*}
&x\equiv\eal=(x^3=\rta_e\,\tau_e,\,x^2=\rta_e\,\bta_e,\,x^1=\rta_e\,\tta_e,\,x^0=\rta_e)\,,\\
&t\equiv(t^3=t^0\tan\alpha_e,\,t^2=t^0\tan\bar{\alpha}_e,\,t^1=t^0\tan\tilde{\alpha}_e,\,t^0)\,,\qquad(t^0\neq0)\,,\\
&s_1\equiv(s_1^3=s_1^0\,\tau_{\sembul},\,s_1^2=s_1^0\,\tadiam,\,s_1^1=s_1^0\,\tau'',\,s_1^0=\tau_{\sembul}\,\tadiam\tau'')\,,\qquad(s_1^0\neq0)\,,\\
&s_2\equiv(s_2^3=s_2^0\,\bta_{\semast},\,s_2^2=s_2^0\,\btadiam,\,s_2^1=s_2^0\,\bta'',\,s_2^0=\bta_{\semast}\,\btadiam\bta'')\,,\qquad(s_2^0\neq0)\,,\\
&s_3\equiv(s_3^3=s_3^0\,\tta_{\sem'},\,s_3^2=s_3^0\,\ttadiam,\,s_3^1=s_3^0\,\tta'',\,s_3^0=\tta_{\sem'}\,\ttadiam\tta'')\,,\qquad(s_3^0\neq0)\,.
\end{align*}
Then, the relations \eqref{crossr} can be put in the forms ($\alpha,\,\beta,\mu,\,\nu=0,1,2,3$; no summation on $i=1,2,3$):
\begin{equation}\label{ki}
{K^i}_{\alpha,\beta,\mu,\nu}\,s_i^\alpha\,s_i^\beta\,x^\mu\,t^\nu=0\,,
\qquad
{K^0}_{\alpha,\beta,\mu,\nu}\,s_i^\alpha\,s_i^\beta\,x^\mu\,t^\nu\neq0\,,
\end{equation}
where the coefficients of the tensors $K^i$ take only the values 0 or $\pm1$ and the only non-vanishing coefficient of $K^0$ is ${K^0}_{0,0,0,0}$\,. Then, it is easy to shown that for all $g_x$ and $g_t$ in $GL(4,\rset)$ then there exists $g\in GL(4,\rset)$ such that
\[
K^i(g\,.\,s_i,g\,.\,s_i,g_x\,.\,x,g_t\,.\,t)=0\,,
\qquad
K^0(g\,.\,s_i,g\,.\,s_i,g_x\,.\,x,g_t\,.\,t)\neq0\,.
\]
In particular, if $g\equiv Id$ and if the $s_i$ are fixed, then the set of localized events $x\equiv{}e_{AL}\in\rset^4_{AL}$ is an orbit of $GL(4, \rset)$ and the set of corresponding events $e_P = \pi(e_{AL})$ is an orbit of the projective group $PGL(4, \rset)$\,.
\par
And then, because the relations \eqref{ki} are homogeneous polynomials, we deduce that $\rset^4_{AL}$ has a projective structure as expected.
\end{proof}
\begin{remark}
The map $\sembb$ defined \emph{locally} by  $\mobf^{-1}\times
\bmobf^{-1}\times\tmobf^{-1}$ on $\rset_P^3$ is the so-called `\emph{soldering map}\footnote{The ancillary emitter $\sem$ can also be considered as the ``soldering'' emitter.}' $\sembb$ of Ehresmann defined on  $P\rset^3=\rset_P^3\cup P\rset^2$ to the spacetime manifold $\mathcal{M}$\,:
$$\xymatrix  @!=1ex {
P\rset^3\times P\rset^3 \ar[rrr] \ar[dd] &&& \ar[ddlll]^{\quad\sembb\underset{loc.}{\simeq}{\mobf}^{-1}\times{\bmobf}^{-1}\times{\tmobf}^{-1}} P\rset^3\\
&&&\\
\hspace{4em}\mathcal{M}\underset{loc.}{\simeq}P\rset^3&&&
}$$
And the set of \textit{homogeneous} equations $K^i(s_i,s_i,t,x)=0$ \textit{when the $s_i$ are fixed} defines leaves in the trivial bundle $P\rset^3\times P\rset^3$. After reduction of the bundle $\rset^4\times\rset^4\to\mathcal{M}$ to this projective bundle, the \emph{projective Cartan connection} in the sense of \emph{Ehresmann} \cite{ehres49} is defined as the differential $dK\equiv dK^1\times dK^2\times dK^3$ with respect to the vertical variables $v\simeq(\tan\alpha_e,\tan\bar{\alpha}_e,\tan\tilde{\alpha}_e)$ and the horizontal variables $e_P$\,; and thus, the tangent spaces of these horizontal leaves are the \emph{annihilators/contact elements} of $dK$\,. 
\end{remark}
\begin{remark}
\label{ehrescartan}
Also, as a result, the spacetime manifold can be considered as a `\emph{generalized Cartan space}' which is endowed with both 1) a `\emph{projective Cartan connection}' (of dimension 4) providing a \emph{local} projective structure, and 2) a compatible (pseudo-)Riemannian structure viewed for instance as a horizontal section in the four dimensional anisotropic grid.
\par
Also, we can eventually define a Ehresmann connection providing a horizontal/vertical splitting in the tangent space of the principal bundle of projective frames of the spacetime manifold. And then, once this Ehresmann connection is given, we can define from this splitting a projective Cartan connection\footnote{We can notice that projective Cartan connections differ from Ehresmann connections which are projector fields (in principal bundles) what projective Cartan connections are not; the word `projective' not referring to a projection in a vector space but to the projective geometry/frames. Also, these two connections differ from the notion of Cartan connection in the sense of Ehresmann which is associated with the definition of the soldering map.} which can be viewed as the infinitesimal changes of the projective frames with respect to themselves.
\end{remark}
\section{RLSs in (3+1)-dimensional spacetimes\label{31rls}}
We need similarly four main emitters $\mathcal{E}$, $\overline{\mathcal{E}}$, $\widetilde{\mathcal{E}}$, $\widehat{\mathcal{E}}$ providing a RPS and, again, one ancillary emitter $\mathcal{S}$ emitting its time coordinates and its own time stamp $\rta$ from a clock to get a RLS.

\subsection{The causal structure of the RPS}
The protocol becomes more and more complex to implement. Indeed, sixteen time stamps are needed to provide to the users their positions in a given system of reference. These positions are obtained from the knowledge the users acquire of their own positions and those of the main emitters both in the emission grid and in the system of reference; and this, thanks to the ephemerides that the emitters upload to the users and the auto-locating process. The causal structure of the RPS is the following (Figure~\ref{rps31} and Table~\ref{tab31}): 
\begin{figure}[!ht]
\caption{\label{rps31}The causal structure of the RPS in a (3+1)-dimensional spacetime.}
$$\xy\xymatrix @!=1em {
&&& U_r &&&\\
E' \ar[rrru] && \be' \ar[ur]&& \te' \ar[lu]&&\he'\ar[lllu]\\
&&&&&&\\
E''\ar[uurr] && \be''\ar[uurr] && \te''\ar[uull]&&\he''\ar[uull]\\
**[r]E'''\ar[uuurrrr]&& \be'''\ar[uuurrrr]&& \te'''\ar[uuullll]&&\he'''\ar[uuullll]\\
**[r]E^{(iv)}\ar[uuuurrrrrr]&&\be^{(iv)}\ar[uuuull]&& \te^{(iv)}\ar[uuuurr]&&\he^{(iv)}\ar[uuuullllll]
}
\endxy
$$
\end{figure}
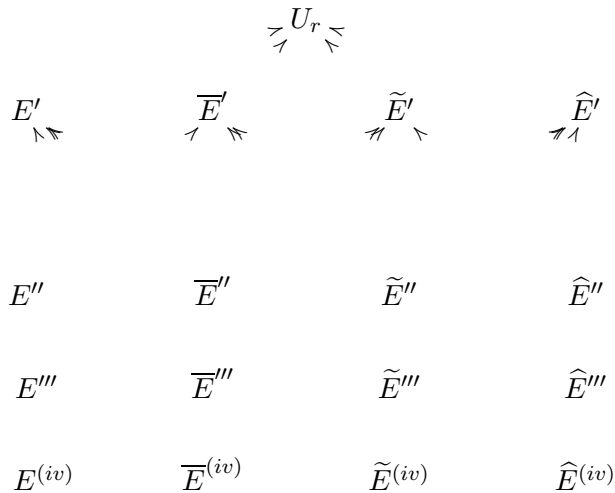
\begin{table}[!ht]
\centering
\caption{\label{tab31}
The events $E'$, $\be'$, $\te'$ and $\he'$ and their broadcast time stamps received at $U_r$.}
\begin{tabular}{ccc}
\hline \hline
Event & \parbox[c]{7em}{\centering broadcasts\\ time stamps}\\ 
\hline 
$E'$& $(\tau'=\tau_{E'},\bta^{(iv)},\tta''',\hta^{(iv)})$ \\
$\be'$& $(\tau'',\bta'=\bta_{\be'},\tta'',\hta''')$ \\
$\te'$& $(\tau''',\bta'',\tta'=\tta_{\te'},\hta'')$ \\
$\he'$& $(\tau^{(iv)},\bta''',\tta^{(iv)},\hta'=\hta_{\he'})$ \\
\hline 
\end{tabular}
\end{table}
\newpage
The position in the emission grid $\rset^4_P$ of the user at $U_r$ is $(\tau',\bta',\tta',\hta')$.
\subsection{The description of the RLS}
As in the $(2+1)$-dimensional case, we need a system of light echoes  associated with each privileged emitter, each linked to an event of reception on the user's worldline. Again, we denote by $\memch$
the system of light echoes for the privileged emitter $\mem$ with $E_p$ as \textit{primary} event. But now, contrarily to the $(2+1)$-dimensional case, we must consider celestial spheres rather than celestial circles. And then, we have again sky mappings from the past null cones directions at the primary events to the ``bright'' points on the associated celestial spheres. Nevertheless, we have only homeomorphisms between hemispheres with half of their boundaries and $\prd$\,. Thus, a problem occurs a priori in this $(3+1)$-dimensional case because we have two disjoint hemispheres for each celestial sphere. And then, consecutive problems appear for the choice and the realization of these hemispheres in the localizing protocol. However, we show in the sequel this problem vanishes completely when considering the full set of echoing systems and the particular hemispheres implementations we present for the emitters. For, we need first the following definition for the determination of the first emission coordinate $\tau_e$\,.
\begin{definition}
\label{echsys31}
\textbf{The echoing system $\memch$} --
The echoing system $\mem{ch}$ associated with the \emph{privileged} emitter $\mem$  is based on the following features (see Figure~\ref{orgedisk}): 
\begin{itemize}
\item one \emph{primary} event $E_p$ with its celestial sphere $\cel_{E_p}$,
\item four \emph{secondary} events $\bebul$, $\tebul$, $\hebul$ with the \emph{ancillary} event $\sembul$, associated respectively with the canonical projective points $[\infty,0]_{E_p}$, $[0,\infty]_{E_p}$, $[0,0]_{E_p}$ and $[1,1]_{E_p}$ of the projective frame $\F_{E_p}$ defining one specific hemisphere of the celestial sphere $\cel_{E_p}$ (Figure~\ref{orgedisk}),
\item one \emph{ternary} events $\esharp$ for $\bebul$\,, two \emph{ternary} events  $\ediam$ and $\bediamd$ for $\tebul$, two \emph{ternary} events  $E''$ and $\besd$ for $\hebul$ and none for $\sembul$\,,
\item two compasses on the specific hemisphere of $\cel_{E_p}$ defined above with a moving origin \emph{anchored} on the projective point $[1,1]_{E_p}$ associated with $\sembul$,
\item one event of reception $U_r\in\wlu$ at which all the data are collected and sent by the emitter $\mem$.
\end{itemize}
\end{definition}
Then, we have the following hierarchy of events in the four different echoing systems $\memch$, $\bemch$, $\temch$ and $\hemch$:
\begin{itemize}
\item Four \textit{primary} events $E_p$, $\be_p$, $\te_p$ and $\he_p$, each with three \textit{secondary} events and one \textit{ancillary} event (Table~\ref{corrproj}):
\newpage 
\begin{table}[ht]
\centering
\begin{tabular}{cccccc}
\hline\hline
\rule[-11pt]{0pt}{28pt}\parbox[c]{4em}{\centering Echoing\\ system}
&\parbox[c]{4em}{\centering Primary\\ event}& $[\infty,0]$ &$[0,\infty]$  & $[0,0]$ & $[1,1]$ \\
\hline
$\memch$&$E_p\in\wl$&$\bebul$&$\tebul$&$\hebul$&$\sembul$\\
$\bemch$&$\be_p\in\wlb$&$\tedag$&$\hedag$&$\edag$&$\semdag$\\
$\temch$&$\te_p\in\wlt$&$\heast$&$\east$&$\beast$&$\semast$\\
$\hemch$&$\he_p\in\wlh$&$E'$&$\be'$&$\te'$&$\sem'$\\
\hline
\end{tabular}
\caption{\label{corrproj}The four primary events and their secondary/ancillary events with their corresponding projective points on the celestial hemispheres $\cel_{E_p}$, $\cel_{\be_p}$, $\cel_{\te_p}$, and $\cel_{\he_p}$ homeomorphic to $P\rset^2$.}
\end{table}
\item Four horismotic relations $E_p\to U_r$, $\be_p\to \bu_r$, $\te_p\to \tu_r$ and $\he_p\to \hu_r$, where the chronologically ordered events of reception $U_r$, $\bu_r$, $\tu_r$ and $\hu_r$ belong to the user worldline $\wlu$\,.
\item One or two (\textit{normal} and \textit{shifted}) \textit{ternary} events by \textit{secondary} event except for the ancillary event:
\begin{subequations}\label{ev24}
\begin{align}
&\memch:& \bebul&: \esharp, & \tebul&: \ediam,\,\bediamd\,, & \hebul&: E'',\,\besd\,,\\
&\bemch:& \tedag&: \besharp, & \hedag&: \bediam,\,\tediamd\,, & \edag&: \be'',\,\tesd\,,\\
&\temch:& \heast&: \tesharp, & \east&: \tediam,\,\hediamd\,, & \beast&:\te'',\,\hesd\,,\\
&\hemch:& E'&: \hesharp, & \be'&: \hediam,\,\ediamd\,, & \te'&: \he'',\,\esd\,,
\end{align}
\end{subequations}
\item Two events associated with the projective points $[\infty,0]$ and $[0,\infty]$ define the equatorial circle dividing the celestial sphere $S^2\simeq P\rset^2\#P\rset^2$ in two celestial hemispheres $\cel$ which are identified to a unique projective space $P\rset^2$. In other words, the directions of propagation of the light rays detected as bright points on the hemispheres are not considered. This could be a problem a priori. Actually, this difficulty is completely canceled out from the operating principles of RLSs as we will see in the sequel.
\item Two compasses on each celestial hemisphere $\cel_{E_p}$, $\cel_{\be_p}$, $\cel_{\te_p}$ and $\cel_{\he_p}$ with a common moving origin for angle measurements \textit{anchored} on the projective point $[1,1]$.
\item We recall that $\sem$ broadcasts as a particular user its own emission coordinates $(\tau_S,\bta_S,\tta_S,\bta_S)$  obtained from the positioning system for all $S\in\wl^\sem$. It broadcasts also all along $\wl^\sem$ its own time coordinate denoted again by~$\rta$\,.
\end{itemize}
\subsection{The causal structure of the RLS}
We represent below (Figures~\ref{echoingmemch}--\ref{fourS} and Tables~\ref{tabechoingmemch}--\ref{tabbisechoingmemch}) only the causal structure for the echoing system $\memch$; the other echoing system $\bemch$, $\temch$ and $\hemch$ can be easily obtained from the symbolic substitutions deduced from Table~\ref{corrproj} and \eqref{ev24}.
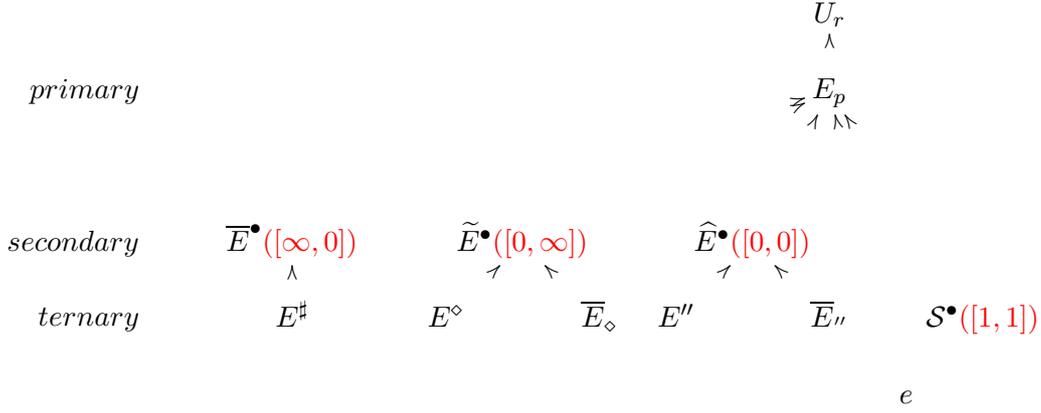
\begin{figure}[!ht]
\centering
$$\xymatrix  @!=1ex {
&&&&&&&&&U_r&&\\
%&&&&&&&&&&&\\
*!R{primary}&&&&&&&&&E_p\ar[u]&&\\
&&&&&&&&&&&\\
*!R{secondary}&&\bebul\txtred{([\infty,0])}\ar[uurrrrrrr]&&&\tebul\txtred{([0,\infty])}\ar[uurrrr] &&&\hebul\txtred{([0,0])}\ar[uur] &&&&\\
*!R{ternary}&&\esharp\ar[u]&&\ediam\ar[ur]&&\bediamd\ar[ul] &E''\ar[ur]&&\besd\ar[ul]&&\sembul\txtred{([1,1])}\ar[uuull]\\
&&&&&&&&&&e\ar[uuuul]&\\
}$$
\vskip-1ex
\textit{$\esharp$, $\ediam$ and $E''$ chronologically ordered}, and \textit{$\bediamd$ and $\besd$ chronologically ordered.}
\smallskip
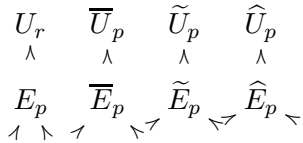
\caption{\label{echoingmemch}The echoing system $\memch$.}
\end{figure}
\begin{table}[!ht]
\centering
\caption{\label{tabechoingmemch}The secondary/ancillary events and their broadcast time stamps in the $\memch$ system.}
\begin{tabular}{ccc}
\hline\hline
Event&\parbox[c]{7em}{\centering broadcasts\\ time stamps}&received at\\
\hline
$\bebul$&$(\tasha,\btabul=\bta_{\bebul})$&$E_p$\\
$\tebul$&$(\tadiam,\btadiamd)$&$E_p$\\
$\hebul$&$(\tau'',\btasd)$&$E_p$\\
$\sembul$&$(\tau_{\sembul},\bta_{\sembul})$&$E_p$\\
\hline
\end{tabular}
\end{table} 
\begin{table}[!ht]
\centering
\caption{\label{tabbisechoingmemch}The ternary events and their broadcast time stamps in the $\memch$ system.}
\begin{tabular}{ccc}
\hline\hline
Event&\parbox[c]{7em}{\centering broadcasts\\ time stamps}&received at\\
\hline
$\esharp$&$\tasha$&$\bebul$\\
$\ediam$&$\tadiam$&$\tebul$\\
$\bediamd$&$\btadiamd$&$\tebul$\\
$E''$&$\tau''$&$\hebul$\\
$\besd$&$\btasd$&$\hebul$\\
\hline
\end{tabular}
\end{table}
\begin{figure}[!ht]
$$\xymatrix  @!=1ex {
&U_r&\bu_p&\tu_p&\hu_p&&
\\
&E_p\ar[u]&\be_p\ar[u]&\te_p\ar[u]&\he_p\ar[u]&&
\\
&&&&&&
\\
e\ar[uur]\ar[uurr]\ar[uurrr]\ar[uurrrr]&&\sembul\ar[uul]&\ll&\semdag\ar[uull]&\ll& \semast\ar[uulll]&\ll&\sem'\ar[uullll]
}$$
\caption{\label{fourS}The causal structure for the four echoing systems $\memch$, $\bemch$, $\temch$ and $\hemch$ with the event~$e$. The chronological order between $\sembul$, $\semdag$, $\semast$ and $\sem'$ belonging to $\wl^\sem$ can be different.}
\end{figure}
\newpage
\subsection{The projective frame, the time stamps correspondence and the consistency}
We consider the projective frame at the primary event $E_p$ and the time stamps correspondence associated with the change of projective frame on $\cel_{E_p}$ (Table~\ref{chgeproj} and Figure~\ref{orgedisk}). Obviously, the other correspondences and changes of projective frames can be deduced in the same way for the three other primary events. From, we obtain four corresponding  pairs of time coordinates for $e$ in the four celestial hemispheres (Table~\ref{eighttim}).
%\par 
%

\begin{table}[ht]
\centering
\caption{\label{chgeproj}The change of projective frame and the corresponding events.}
\begin{tabular}{ccccc}
\hline \hline
Event &\mbox{\qquad}& $\F_{E_p}$ &\mbox{\qquad}& $\F^\tau_{E_p}$\\ 
\hline 
$e$&& $[\tan\alpha_e,\tan\beta_e]$&& $[\tabul_e,\btabul_e]$\\ 
$\bebul$&& $[\infty,0]$ && $[\tasha,\btabul=\bta_{\bebul}]$\\
$\tebul$&& $[0,\infty]$ && $[\tadiam,\btadiamd]$\\ 
$\hebul$&& $[0,0]$ && $[\tau'',\btasd]$\\ 
$\sembul$&& $[1,1]$ && $[\tau_{\sembul},\bta_{\sembul}]$\\
\hline 
\end{tabular}
\end{table}
\begin{table}[ht]
\centering
\caption{\label{eighttim}The pairs of time coordinates for $e$ deduced in the four celestial hemispheres.}
\begin{tabular}{ccc}
\hline \hline
\parbox[c]{5em}{\centering Celestial\\
hemisphere}&\parbox[c]{5em}{\centering Time coordinates for $e$}\rule[-16pt]{0pt}{38pt}&\\ 
\hline 
$\cel_{E_p}$&$(\tabul_e,\btabul_e)$&\\
$\cel_{\be_p}$&$(\btadag_e,\ttadag_e)$&\\
$\cel_{\te_p}$&$(\ttaast_e,\htaast_e)$&\\
$\cel_{\he_p}$\rule[-7pt]{0pt}{0pt}&$(\hta'_e,\tau'_e)$&\\
\hline 
\end{tabular}
\end{table}
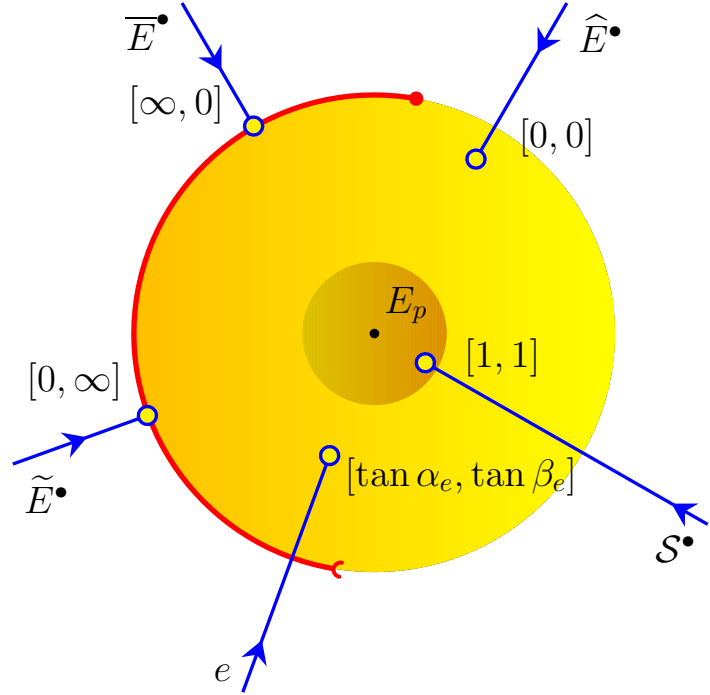
\begin{figure}[!ht]
\centering
\caption{\label{orgedisk}The projective disk on the celestial hemisphere $\cel_{E_p}$ centered at $E_p$  and the four canonical projective points and the corresponding projective point for $e$\,.}
\begin{tikzpicture}[scale=.9]
\fill[
left color=yellow!50!orange,
right color=yellow, 
shading=axis,
fill opacity=1
](0,0) circle [radius=100pt];
\fill[
left color=yellow,
right color=red,
shading=axis,
opacity=.175
](0,0) circle [radius=30pt];
\fill[color=black,](0,0) node[anchor=south west]{\Large$E_p$} circle [radius=2pt];
\draw[color=red,line width=2pt,opacity=1](0,0) +(80:100pt) arc [radius=100pt,start angle=80,delta angle=180];
\fill[color=red](0,0) +(80:100pt) circle [radius=3pt];
\draw[color=red,line width=1.5pt,line cap=round](0,0) +(262:97pt) arc [radius=3pt,start angle=80,delta angle=180];
\draw[blue,line width=1.25pt,arrows=-{Stealth[width=10pt,sep=20pt]Circle[length=8pt,open,blue,fill=yellow,line width=1.25pt]}]
(0,0) (120:160pt) node[black,anchor=north east]{\Large$\bebul$} -- (120:96pt) node[black,anchor=south east]{\Large$[\infty,0]$\,\,\,\,};
\draw[blue,line width=1.25pt,arrows=-{Stealth[width=10pt,sep=20pt]Circle[length=8pt,open,blue,fill=yellow,line width=1.25pt]}]
(0,0) (200:160pt) node[black,anchor=north west]{\Large$\tebul$} -- (200:96pt) node[black,anchor=south east]{\Large$[0,\infty]$\,\,\,\,};
\draw[blue,line width=1.25pt,arrows=-{Stealth[width=10pt,sep=70pt]Circle[length=8pt,open,blue,fill=yellow,line width=1.25pt]}]
(0,0) (250:160pt) node[black,anchor=south east]{\Large$e$} -- (250:50pt) node[black,anchor=north west]{\Large$[\tan\alpha_e,\tan\beta_e]$};
\draw[blue,line width=1.25pt,arrows=-{Stealth[width=10pt,sep=102pt]Circle[length=8pt,open,blue,fill=yellow,line width=1.25pt]}]
(0,0) (330:160pt) node[black,anchor=north east]{\Large$\sembul$} -- (330:20pt) node[black,anchor= west]{\Large\quad$[1,1]$};
\draw[blue,line width=1.25pt,arrows=-{Stealth[width=10pt,sep=42pt]Circle[length=8pt,open,blue,fill=yellow,line width=1.25pt]}]
(0,0) (60:160pt)node[black,anchor=north west]{\Large$\hebul$} --  (60:80pt) node[black,anchor=south west]{\Large\quad$[0,0]$};
\end{tikzpicture}
\end{figure}
\newpage
Then, the change of projective frame on the celestial hemisphere gives the following relation:
\begin{equation}\label{chprof31}
\begin{pmatrix}
a&d&g\\
b&e&h\\
c&f&k
\end{pmatrix}
\begin{pmatrix}
1&0&0&1&\rho\,\tan\alpha_e\\
0&1&0&1&\rho\,\tan\beta_e\\
0&0&1&1&\rho\,
\end{pmatrix}=
\begin{pmatrix}
u\,\tasha&v\,\tadiam&w\,\tau''&(u+v+w)\,\tau_{\sembul}&r\,\tabul_e\\
u\,\btabul&v\,\btadiamd&w\,\btasd&(u+v+w)\,\bta_{\sembul}&r\,\btabul_e\\
u&v&w&(u+v+w)&r
\end{pmatrix},
\end{equation}
where $\rho\,\,u\,v\,w\,r\,(u+v+w)\neq0$, and where the determinant of the square matrix on the \textit{l.h.s.} of this equality must be non-vanishing.
Then, we deduce that
\begin{equation}\label{chprof31b}
\begin{pmatrix}
u\,\tasha&v\,\tadiam&w\,\tau''\\
u\,\btabul&v\,\btadiamd&w\,\btasd\\
u&v&w
\end{pmatrix}
\begin{pmatrix}
1&\rho\,\tan\alpha_e\\
1&\rho\,\tan\beta_e\\
1&\rho\,
\end{pmatrix}=
r\,(u+v+w)
\begin{pmatrix}
\tau_{\sembul}&\tabul_e\\
\bta_{\sembul}&\btabul_e\\
1&1
\end{pmatrix},
\end{equation}
and we can take in addition $u+v+w=r=1$.
\par
Obviously, we obtain  from the three other echoing systems three other similar systems of equations for the other six unknown time coordinates given in Table~\ref{eighttim} for $e$.
\par
Now, besides, we have necessarily the relations:
\begin{equation}\label{taueight}
\tabul_e=\tau'_e\,,
\qquad
\btabul_e=\btadag_e\,,
\qquad
\ttaast_e=\ttadag_e\,,
\qquad
\htaast_e=\hta'_e\,.
\end{equation}
Indeed, if one of these four precedent equalities is not satisfied, then it means that if the event $e$, the worldlines of the main emitters and the ancillary one are fixed, then, at least one time stamp among the eight can vary. And then, one of the eight angles on the four celestial hemispheres  necessarily can vary as well. But then, it would mean that the position of the event $e$ seen on the celestial hemispheres of the four main emitters can vary arbitrarily whenever $e$ is fixed. In other words, $e$ might have more than one corresponding ``bright'' point on each celestial hemisphere; and, in particular, because we have continuous functions, then it might correspond to $e$, in particular, a connected ``bright'' line on one of the four celestial hemispheres. This would involve necessarily the existence of more than one and only one horismotic relation `\,$\to$.' This situation can be encountered in the case of the existence of conjugate points for light-like geodesics for instance in Riemannian manifolds. Then, considering only one horismos, the relations \eqref{taueight} must be satisfied.
\par
Then, we obtain:
\begin{lemma}
\label{lemauto31}
Let $\fourtor\equiv(\pru)^4$ be the 4-torus. Then,
the RLS provides a map
\begin{equation}\label{autol4}
\mathfrak{M}_4:(\tan\alpha_e,\tan\bar{\alpha}_e,\tan\tilde{\alpha}_e,\tan\hat{\alpha}_e)\in\fourtor
\longrightarrow(\tau_e,\bta_e,\tta_e,\hta_e)\in\fourtor
\end{equation}
which is an automorphism.
\end{lemma}
\begin{proof}
This lemma can be easily proved simply by solving systems of equations like \eqref{chprof31b} but we indicate interesting intermediate homogeneous equations in the computations.
From the relations \eqref{taueight} and the equations at each primary event such as the equations \eqref{chprof31b} at $E_p$, we deduce that there are four linear relations between the ``$\tan\alpha$'' and the ``$\tan\beta$\,.'' And then, it can be shown that we obtain four M{\"o}bius relations linking the four $\tan\alpha$'s to the four time coordinates $\tau_e$, $\bta_e$, $\tta_e$ and $\hta_e$ of $e$ generalizing the situation encountered in the precedent (2+1)-dimensional case. 
\par
More precisely, considering the primary event $E_p$\,, we obtain the equations \eqref{chprof31b}. At the other primary event $\te_p$, we obtain the similar following relations ($\nu\,\neq0)$:
\begin{equation}\label{chprof31t}
\begin{pmatrix}
p\,\ttasha&q\,\ttadiam&m\,\tta''\\
p\,\htaast&q\,\htadiamd&m\,\htasd\\
p&q&m
\end{pmatrix}
\begin{pmatrix}
1&\nu\,\tan\talpha_e\\
1&\nu\,\tan\tbeta_e\\
1&\nu\,
\end{pmatrix}=
n\,(p+q+m)
\begin{pmatrix}
\tta_{\semast}&\ttaast_e\\
\hta_{\semast}&\htaast_e\\
1&1
\end{pmatrix},
\end{equation}
where again we can impose the relations $n=p+q+m=1$\,. Then, from now and throughout, we set:
\begin{equation}\label{taueightfour}
\tau_e\equiv\tabul_e=\tau'_e\,,
\qquad
\bta_e\equiv\btabul_e=\btadag_e\,,
\qquad
\tta_e\equiv\ttaast_e=\ttadag_e\,,
\qquad
\hta_e\equiv\htaast_e=\hta'_e\,.
\end{equation}
And then, it can be shown that among the relations \eqref{chprof31b} and \eqref{chprof31t}, those depending explicitly on the time stamps can be put in the following forms ($p,q=1,\ldots,4$ and $\mu,\nu=0,1,\ldots,4$):
\begin{subequations}\label{kete31}
\begin{align}
&K_{pq\mu\nu}\,s^p\,S^q\,x^\mu\,t^\nu=0\,,
&
\overline{K}_{pq\mu\nu}\,\bar{s}^p\,S^q\,x^\mu\,t^\nu=0\,,
\\
&\widetilde{K}_{pq\mu\nu}\,\tilde{s}^p\,\widetilde{S}^q\,x^\mu\,t^\nu=0\,,
&
\widehat{K}_{pq\mu\nu}\,\hat{s}^p\,\widetilde{S}^q\,x^\mu\,t^\nu=0\,,
\end{align}
\end{subequations}
where
%
%\begin{subequations}
%\begin{align}
%x&\equiv\left(\rta_e\,\tau_e,\rta_e\,\bta_e,\rta_e\,\tta_e,\rta_e\,\hta_e,\rta_e\right),&
%t&\equiv\left(\mu\,\tan\alpha,\mu\,\tan\beta,\mu\,\tan\talpha,\mu\,\tan\tbeta,\mu\right), 
%\\
%s&\equiv\left(s^0\,\tasha,s^0\,\tadiam,s^0\,\tau'',s^0=\tasha\,\tadiam\,\tau''\right), & \bar{s}&\equiv\left(\bar{s}^0\,\btabul,\bar{s}^0\,\btadiamd,\bar{s}^0\,
%\btasd,\bar{s}^0=\btabul\,\btadiamd\,\btasd\right), \\
%\tilde{s}&\equiv\left(\tilde{s}^0\,\ttasha,\tilde{s}^0\,\ttadiam,
%\tilde{s}^0\,\tta'',\tilde{s}^0=\ttasha\,\ttadiam\,\tta''\right), & \hat{s}&\equiv\left(\hat{s}^0\,\htaast,\hat{s}^0\,\htadiamd,\hat{s}^0\,
%\htasd,\hat{s}^0=\htaast\,\htadiamd\,\htasd\right),
%\end{align}
%\end{subequations}
%
\begin{subequations}
\begin{align}
x&\equiv\left(\rta_e\,\tau_e,\rta_e\,\bta_e,\rta_e\,\tta_e,\rta_e\,\hta_e,\rta_e\right),&
t&\equiv\left(\mu\,\tan\alpha,\mu\,\tan\beta,\mu\,\tan\talpha,\mu\,\tan\tbeta,\mu\right), 
\\
s&\equiv\left(\tasha,\tadiam,\tau'',s^0\right), & \bar{s}&\equiv\left(\btabul,\btadiamd,
\btasd,\bar{s}^0=s^0\right), \\
\tilde{s}&\equiv\left(\ttasha,\ttadiam,
\tta'',\tilde{s}^0\right), & \hat{s}&\equiv\left(\htaast,\htadiamd,
\htasd,\hat{s}^0=\tilde{s}^0\right),
\end{align}
\end{subequations}
and
\begin{subequations}
\begin{align}
&S\equiv\left(
\begin{vmatrix}
\tau_{\sembul}& \tadiam & \tau'' \\ 
\bta_{\sembul}& \btadiamd & \btasd \\ 
1 & 1 & 1
\end{vmatrix},
\begin{vmatrix}
\tasha& \tau_{\sembul} & \tau'' \\ 
 \btabul& \bta_{\sembul} & \btasd \\ 
1 & 1 & 1
\end{vmatrix},
\begin{vmatrix}
\tasha& \tadiam&\tau_{\sembul}  \\ 
 \btabul& \btadiamd&\bta_{\sembul} \\ 
1 & 1 & 1
\end{vmatrix},
\begin{vmatrix}
\tasha& \tadiam&\tau''  \\ 
 \btabul& \btadiamd&\bta'' \\ 
1 & 1 & 1
\end{vmatrix}
\right),
\\
&\widetilde{S}\equiv\left(
\begin{vmatrix}
\tta_{\semast}& \ttadiam & \tta'' \\ 
\hta_{\semast}& \htadiamd & \htasd \\ 
1 & 1 & 1
\end{vmatrix},
\begin{vmatrix}
\ttasha& \tta_{\semast} & \tta'' \\ 
\htaast& \hta_{\semast} & \htasd \\ 
1 & 1 & 1
\end{vmatrix},
\begin{vmatrix}
\ttasha& \ttadiam&\tta_{\semast}  \\ 
\htaast& \htadiamd&\hta_{\semast} \\ 
1 & 1 & 1
\end{vmatrix},
\begin{vmatrix}
\ttasha& \ttadiam&\tta''  \\ 
\htaast& \htadiamd&\htasd \\ 
1 & 1 & 1
\end{vmatrix}
\right),
\end{align}
\end{subequations}
where $s^0$ and $\tilde{s}^0$ are non-vanishing arbitrary coefficients and $\mu\neq0$\,.
Hence, we have four homogeneous algebraic equations linking the vectors $t$ and $x$\,.
Obviously, we have also four other similar homogeneous equations for $x$ and an other $t'\simeq(\tan\balpha,\tan\bbeta,\tan\halpha,\tan\hbeta)$ deduced from the echoing systems at the other two primary events $\be_p$ and $\he_p$:
\begin{subequations}\label{kbehe31}
\begin{align}
&K_{pq\mu\nu}\,s'^p\,\overline{S}^q\,x^\mu\,t'^\nu=0\,,
&
\overline{K}_{pq\mu\nu}\,\bar{s}'^p\,\overline{S}^q\,x^\mu\,t'^\nu=0\,,
\\
&\widetilde{K}_{pq\mu\nu}\,\tilde{s}'^p\,\widehat{S}^q\,x^\mu\,t'^\nu=0\,,
&
\widehat{K}_{pq\mu\nu}\,\hat{s}'^p\,\widehat{S}^q\,x^\mu\,t'^\nu=0\,,
\end{align}
\end{subequations}
and where $s'^0=\bar{s}'^0$ and $\tilde{s}'^0=\hat{s}'^0$ are non-vanishing arbitrary coefficients.
And then, because $x$ is determined completely from the equations \eqref{kete31}, the equations  \eqref{kbehe31} are linearly depending on the equations \eqref{kete31} which involves that we have linear relations between the two set of ``angles'' $t\simeq(\tan\alpha,\tan\beta,\tan\talpha,\tan\tbeta)$ and $t'\simeq(\tan\balpha,\tan\bbeta,\tan\halpha,\tan\hbeta)$. Hence, taking linear combinations of the systems of equations \eqref{kbehe31} and \eqref{kete31} and taking into account also the remaining equations in \eqref{chprof31b} and \eqref{chprof31t} not depending on the time stamps we can deduce a system of four homogeneous equations linking $(\tan\alpha,\tan\balpha,\tan\talpha,\tan\halpha)$ and $x$:
\begin{subequations}
\begin{equation}\label{k1234}
H^1_{\mu,\nu}\,x^\mu\,t_1^\nu=0\,,
\qquad
H^2_{\mu,\nu}\,x^\mu\,t_2^\nu=0\,,
\qquad
H^3_{\mu,\nu}\,x^\mu\,t_3^\nu=0\,,
\qquad
H^4_{\mu,\nu}\,x^\mu\,t_4^\nu=0\,,
\end{equation}
where $s^0$, $\tilde{s}^0$, $s'^0$ and  $\tilde{s}'^0$ do not intervene anymore and 
\begin{align}
t_1&\equiv\left(\mu_1\,\tan\alpha,0,0,0,\mu_1\right), &
t_2&\equiv\left(0,\mu_2\,\tan\balpha,0,0,\mu_2\right), \\
t_3&\equiv\left(0,0,\mu_3\,\tan\talpha,0,\mu_3\right), &
t_4&\equiv\left(0,0,0,\mu_4\,\tan\halpha,\mu_4\right).
\end{align}
\end{subequations}
And these equations \eqref{k1234} determine univocally $x$ up to the time coordinate $\rta_e$, \ie, $(\tau_e,\bta_e,\tta_e,\hta_e)$. Moreover, the equations \eqref{k1234} are other expressions for M{\"o}bius transformations between each given angle and a linear combination of  $(\tau_e,\bta_e,\tta_e,\hta_e)$; hence the result for an automorphism on the 4-torus.
\end{proof}
\begin{remark}
We can notice from the Lemma~\ref{lemauto31} that we obtain the time coordinates $(\tau_e,\bta_e,\tta_e,\hta_e)$ for $e$ from only two echoing systems, \eg, at $E_p$ and $\te_p$ with the four ``angles'' $(\tan\alpha,\tan\beta,\tan\talpha,\tan\tbeta)$, or from the four echoing systems at $E_p$, $\be_p$, $\te_p$ and $\he_p$ with the four ``angles'' $(\tan\alpha,\tan\balpha,\tan\talpha,\tan\halpha)$.
\end{remark}
\begin{theorem}\label{consist31}
The localization and the positioning protocols or systems in a (3+1)-dimensional spacetime are consistent.
\end{theorem}
\begin{proof}
Obvious, because 1) the RLS in the (3+1)-dimensional case has a causal structure which can be decomposed in four causal sub-structures each equivalent to the one given for the RLS in the (2+1)-dimensional case, and 2) we need only the ``angles'' $\alpha$ to localize $e$ in each of these sub-systems of localization.
\end{proof}
\subsection{The local projective structure}
\begin{definition}\label{grids31}
We call
\begin{itemize}
\item \emph{Emission grid}, the Euclidean space $\rset^4_P\equiv\rset^4$ of the positioned events $\ep=(\tau_e,\bta_e,\tta_e,\hta_e)$,
\item \emph{Localization (\emph{or} pentametric) grid}, the Euclidean space $\rset^5_L\equiv\rset^4\times\rset^\ast$ of the localized events $\el=(\tau_e,\bta_e,\tta_e,\hta_e,\rta_e)$ where $\rta_e$ is provided by the ancillary emitter $\sem$ by \emph{identification} from the horismotic relation $S^L\to\ep$ $(S^L\in\wl^\sem)$ or the `\emph{message function}' \cite{Woodhouse73} $f_\sem^-:\rset^4_P\longrightarrow\wl^\sem$, \ie, the time stamp $\rta^L$ broadcast by $\sem$ at $S^L$ is such that $\rta^L\equiv\rta_e$, and
\item \emph{Anisotropic localization (\emph{or} pentametric) grid}, the Euclidean space $\rset^5_{AL}\equiv\rset^4\times\rset^\ast$ of  events $\eal=(\rta_e\tau_e,\rta_e\bta_e,\rta_e\tta_e,\rta_e\hta_e,\rta_e)$.
\end{itemize}
\end{definition}
\begin{definition}
We denote by $I:\rset_L^5\longrightarrow
\rset_{AL}^5$\,, the bijective map such that $I(\el)=\eal$. And we denote by $\pi:\rset_{AL}^5\longrightarrow
\rset_{P}^4$\,, the submersion such that $\pi(\eal)=\ep$.
\end{definition}
Let $g$ be an element of $GL(5,\rset)$ such that $g\,.\,\eal=\eal'$\,. And thus, $GL(5,\rset)$ acts linearly on $\rset_{AL}^5$. Then, the action of $GL(5,\rset)$ on $\rset_{L}^5$ and $\rset_P^4$ defines homographies (\ie, conformal transformations):
\begin{subequations}\label{hombis}
\begin{align}
&\ep'=\left(\frac{A\,.\,\ep+b}{c\,.\,\ep+\mu}\right),
\qquad
g\equiv\begin{pmatrix}
A& b \\ 
 {}^ tc&\mu
\end{pmatrix},\label{epabis}
\\
&\rta'_e=\rta_e(c\,.\,\ep+\mu)\,,
\end{align}
\end{subequations}
where $\mu\in\rset$\,, $(b,\,c)\in(\rset^4)^2$ and $A\in M_{4\times4}(\rset)$\,. Therefore, we obtain:
\begin{theorem}\label{pgl5}
The (3+1)-dimensional spacetime manifold has a \emph{local} 4-dimensional projective structure inherited from its causal structure.
\end{theorem}
\begin{proof}
The proof is similar with the proof of Theorem~\ref{pgl4} but with the systems of homogeneous equations \eqref{kete31} or \eqref{kbehe31} or \eqref{k1234} instead of the system \eqref{ki}.
\end{proof}
\begin{remark}
The map $\sembb$ defined \emph{locally} by  $\mathfrak{M}_4$ on $\rset_P^4$ is the so-called `\emph{soldering map}' $\sembb$ of Ehresmann defined on $P\rset^4=\rset_P^4\cup P\rset^3$ to the spacetime manifold $\mathcal{M}$\,:
$$\xymatrix  @!=1ex {
P\rset^4\times P\rset^4 \ar[rrr] \ar[dd] &&& \ar[ddlll]^{\displaystyle\quad\sembb\underset{loc.}{\simeq}\mathfrak{M}_4} P\rset^4\\
&&&\\
\hspace{4em}\mathcal{M}\underset{loc.}{\simeq}P\rset^4&&&
}$$
And the set of \textit{homogeneous} equations \eqref{kete31} or \eqref{kbehe31} or \eqref{k1234} defines leafs in the trivial bundle $P\rset^4\times P\rset^4$. After reduction of the bundle $\rset^5\times\rset^5\to\mathcal{M}$ to this projective bundle, the \emph{projective Cartan connection} in the sense of \emph{Ehresmann} \cite{ehres49} is defined as the differential $dH\equiv dH^1\times dH^2\times dH^3\times dH^4$ with respect to the vertical variables $v\simeq(\tan\alpha,\tan\balpha,\tan\talpha,$ $\tan\halpha)$ and the horizontal variables $e_P$\,; and thus, the tangent spaces of these horizontal leafs are the \emph{annihilators/contact elements} of $dH$\,. 
\end{remark}
\begin{remark} 
\label{ehrescartan31}
As in the $(2+1)$-dimensional case 
(see Remark~\ref{ehrescartan}), the spacetime manifold can be considered as a `\emph{generalized Cartan space}' which is endowed with both 1) a `\emph{projective Cartan connection}' (of dimension 5) providing a \emph{local} projective structure, and 2) a compatible (pseudo-)Riemannian structure viewed for instance as a horizontal section in the five dimensional anisotropic grid.
\end{remark}
\section{Conclusion}
The previous results are obtained on manifolds of dimension less or equal to four and satisfying only causal axiomatics. This involves only the following assumptions:
\begin{itemize}
\item A finite speed of light from the existence of the horismotic relation `$\to$,'
\item Isotropy (compasses) because the conformal invariance is a common consequence of all the causal axiomatics (Malament's Theorem \cite{Malament77}, Woodhouse's axiomatics \cite{Woodhouse73}, King-Hawking-McCarty's axiomatics \cite{Kingal76} for instance; see also \cite{EhlPiSch:72}).
\item Homogeneity is void of meaning in causal axiomatics.
\end{itemize}
And then, we deduced that 
\begin{itemize}
\item the spacetime manifold has a \textit{local} projective structure in addition to the \textit{global} (pseudo-) \text{Riemannian} structure.
\item The spacetime manifold is a `\textit{generalized Cartan space}' with a `\textit{projective Cartan connection}' (see Remarks~\ref{ehrescartan} and \ref{ehrescartan31}). A forthcoming publication is planned to clarify these aspects.
\item The space and time coordinates are locally transformed by homographies. Indeed, the time stamps $(\tau^\alpha)=(\tau_e,\bta_e,\tta_e,\hta_e)$ $(\alpha=1,\ldots,4)$ are the `emission coordinates.' Then, they define a null co-frame $(d\tau^\alpha)$ also called `null GPS coordinates'  which are \textit{linearly} related to timelike (GPS) coordinates \cite{Rovelli:2002gps} such as the usual space and time coordinates $(u^\alpha)\equiv(x,y,z,t)$\,. Therefore, we obtain transformations similar to the transformations \eqref{epabis} for the space and time coordinates:
\begin{equation}\label{homogspacetime}
 u^\alpha=\left(\frac{U^\alpha_\beta\,u'^\beta+v^\alpha}{w_\mu\,u'^\mu+\rho}\right), \qquad\qquad \alpha\,,\,\beta\,,\,\mu=1,\ldots,4.
\end{equation}
\end{itemize}
Besides, applying by inquisitiveness these projective aspects in astrophysics, we consider a modification of the Newton's law of gravitation by a homography:
\begin{equation}
\vec{F}(\vec{r}_0)\equiv -\,G\,\frac{m_0\,m}{r_0^2}\,\left(\alpha+\beta\,t+\mu\,r_0\right)^2\,
\hat{r}_0\,,
\qquad
\hat{r}_0=\vec{r}_0/r_0\,,
\label{mondnew}
\end{equation}
where we consider that $\vec{F}$, the time $t$ and the radial distance $r_0$ between the punctual masses $m_0$ and $m$ are \textit{evaluated with respect to a frame attached to} $m$\,, and $\alpha$, $\beta$ and $\mu$ are constants.
\par
This modification differs from those investigated in MOND theories which satisfy the so-called \text{Milgrom's} law \cite{Famaey2012}. Also, contrarily to MOND theories, the present modification of Newton's law preserves the action/reaction principle. This modification is based on the notion of \textit{projective tensor} differing from the usual notion of \textit{Euclidean tensor}. 
\par 
We can quote {\'E}.~Cartan on this notion of \textit{Euclidean tensor} \cite[\S23, p.~22]{Cartan1981}. The latter can be considered as a set of numbers $(u^1,u^2,\ldots,u^r)$ brought into coincidence with an other set of numbers $(u'^1,u'^2,\ldots,u'^r)$ by a \textit{linear} transformation $S_o$ corresponding to a \textit{rotation} $R_o$ in a given ``underlying'' Euclidean space $\rset^k$\,. Then, the linear transformation $S_o$ corresponds to another transformation $R_o$ preserving the origin $o$ of $\rset^k$\,. Thus, we obtain tensors at this origin.
\par
Now, a \textit{projective tensor} can be considered as a set of numbers $(v^1,v^2,\ldots,v^r)$ brought into coincidence with an other set of numbers $(v'^1,v'^2,\ldots,v'^r)$ by a \textit{linear} transformation $S_o$ corresponding to an \textit{homography} $H_o$ preserving the origin $o$ ($H_o$ is then a \textit{central collineation} of center $o$) of a given ``underlying'' Euclidean space $\rset^k$\,.
In other words, the equivariance is defined for Euclidean tensors with respect to linear groups of transformations whereas the equivariance for projective tensors is defined with respect to the group of central collineations which is a subgroup of the projective group.
\par
Also,  if we have tensor fields, \ie, tensors at different origins $p$ elements of a manifold (as space of parameters), then, there corresponds fields (or families) of transformations $S_p$ on this manifold associated with fields of rotations $R_{o,p}$ (Euclidean tensor fields) or fields of central collineations $H_{o,p}$ (projective tensor fields) associated with the origin $o$ of the underlying Euclidean space $\rset^k$. Then, the tensor fields are equivariant if and only if the equivariance is satisfied at any point $p$\,. We recognize in this description the structure of a tensor bundle of rank $k$ of which the transition morphims (functions) are the rotation or the collineation fields, the transformation fields $S_p$ are defined from the local trivializations of the bundle and the origin $o$ is an element of the fiber. Moreover, the equivariance of tensor fields is obviously the so-called left-invariance with respect to right actions of structural groups.
\par\medskip
Then, we consider, first, the (non-modified) force of gravitation $\vec{F}$ as a Euclidean vector field with a spherical symmetry with respect to the point $p_0$ where the mass $m_0$ is located. The mass $m$ is at the point $p$ and, as indicated previously, $\vec{F}$ and $\vec{r}_0$ are evaluated with respect to a Euclidean frame attached to $p$\,. Then, clearly, rotating this frame does not change $r_0$ and it rotates in the same way the vector $\hat{r}_0$\,.
\par\medskip
Second, if $\vec{F}$ is modified to be a projective vector field with a spherical symmetry with respect to $p_0$\,, we must proceed as follows. In this projective framework, the central collineation fields $H_{o,p}$ are defined such as at each $p$ they are particular changes of projective frames $\F_p$\,. More precisely, we recall that the projective frames of a projective space of dimension four are defined by six vectors in a vector space of dimension five of which five are linearly independent. Then, projective transformations, \ie, homographies, are the linear, injective transformations in this five-dimensional vector space which is also called the space of \textit{homogeneous} coordinates. Then, central collineations are those projective transformations preserving --up to a multiplicative factor-- a particular, given five-dimensional vector, \ie, the origin $o$ of the five-dimensional fiber. In general, this vector is chosen among the vectors of a given projective frame.
Moreover, $p$ can be kept invariant with respect to central collineations which constitute a subgroup of the group of projective transformations. Indeed, these central collineations can also be viewed as local changes of \textit{inhomogeneous} coordinates centered at $p$\,, \ie, $p$ is the origin of the local system of coordinates. Hence, if $(x,y,z,t)$ are space and time coordinates centered at $p$ such that $p\equiv(0,0,0,0)=(x_p,y_p,z_p,t_p)$, then the changes of coordinates we must consider are given by the homographies  \eqref{homogspacetime}
where $v^\alpha=0$ ($\alpha=1,\ldots,4$), \ie, we have central collineations $H_{p,o}:(u'^\alpha)\equiv(x',y',z',t')\longrightarrow(u^\alpha)\equiv(x,y,z,t)$ such that
\begin{equation}\label{usimp}
u^\alpha=\frac{U^\alpha_\beta\, u'^\beta}{(q+h\,t'+\vec{k}\,.\,\vec{r'}\,)}\,,
\end{equation}
where $\vec{r'}\equiv(u'^i)\equiv(x',y',z')$ ($i,j,\ldots=1,2,3$)\,, $(U^\alpha_\beta)$ is a matrix field, $q$ and $h$ are scalar fields and $\vec{k}$ is a vector field all of them depending on $p$\,. 
\par
In particular, if the time and space splitting of the Newtonian physics is maintained in this change of coordinates, then we must have $U^i_4=U^4_i=0$\,. And then, we deduce in particular that
\begin{equation}\label{elatr}
r=\lambda\,\frac{r'}{(q+h\,t'+\vec{k}\,.\,\vec{r'}\,)}\,,
\qquad\qquad
t=\mu\,\frac{t'}{(q+h\,t'+\vec{k}\,.\,\vec{r'}\,)}\,.
\end{equation}
where $\lambda=\det(U^i_j)$ and $\mu=U^4_4$\,. Also, from \eqref{usimp}, considering that 1) $\vec{F}$ is a projective five-vector with two vanishing components, \ie, $\vec{F}\equiv(F^1,F^2,F^3,0,0)$, and 2) the central collineations are represented (or, are originated from) by the linear transformations $U\equiv(U^a_b)$ ($a,b=1,\ldots,5$) such that $U^a_5=0$\,, $U^5_i=k^i$ ($i=1,\ldots,3$), $U^5_4=h$ and $U^5_5=q$, then we obtain in particular $F^i\equiv\sum_{j=1}^{3} U^i_j\,F'^j$\,.
\par
Then, it is easy to see that \eqref{mondnew} becomes equivariant with respect to these changes of coordinates if and only if we set the necessary but, nevertheless, sufficient condition that the vector field $\vec{k}$ satisfies the relation: $\vec{k}\equiv\sigma\,\vec{r'_0}/r'_0$\,, where $\sigma$ is a scalar field depending on $p$ and $\vec{r'_0}$ is the vector from $p$ to $p_0$ in the new system of coordinates. Indeed, with this condition, we obtain the new Newtonian force:
\[
\vec{F'}(\vec{r'}_0)\equiv -\,G\,\frac{m_0\,m}{{r'}_0^2}\,\left(\alpha'+\beta'\,t'+\mu'\,r'_0\right)^2\,
\hat{r}'_0\,,
\qquad
\hat{r}'_0=\vec{r'}_0/r'_0\,.
\]
More precisely, the equivariance is obtained as soon as $\vec{r'}$ is equal to $\vec{r'_0}$, \ie, when we move along the line joining $p$ and $p_0$, and then, $\vec{F}$ is an equivariant, projective vector field along this line onto which only the Newton's law is experimentally evaluated. Also, we obtain the field of transformations $S_p$ as expected and a justification of the modification \eqref{mondnew} of the Newton's law of gravitation.
\par
Furthermore, $\vec{F}$ is then only equivariant with respect to a subgroup of the central collineations because we must set $\vec{k}\equiv\sigma\,\hat{r}'_0$\,. This is the result of 1) the central symmetry of $\vec{F}$\,, and 2) the Newtonian physics framework with the time and space splitting. Also, it can be noticed that in this Newtonian context, the terms such as $\beta\,t+\mu\,r_0$ in the expression of $\vec{F}$ may sound like a Minkowski inner product and could be the expression of a retarded Newtonian force as there exists retarded fields in electromagnetism.
\par 
In addition, we assume that the centripetal acceleration is a projective object and that it is modified in the same way as the Newton's law of gravitation: $v^2/r_0\longrightarrow  (\alpha+\beta\,t+\mu\,r_0)\,v^2/r_0$\,.
\par\medskip 
Then, if we modify the Newton's law of gravitation with a homography as indicated above preserving the mass distribution $\rho(r)$ to see the relative change between $\rho(r)$ and the radius $r$, then we can deduce the following rotational velocity field:
\begin{equation}\label{rvf}
v(r)\equiv\left(
(1+a\,t+b\,r)\,\frac{M(r)}{r}
\right)^{1/2}
\qquad\text{where}\qquad
M(r)\equiv\int_{0}^{r}u^k\,\rho(u)\,du\,.
\end{equation}
Then, whenever $t=1$ and for different mass distributions $\rho$, we obtain the following qualitative curves if we consider $k=2$ for spherical distributions of mass in \eqref{rvf} (Figure~\ref{curverot} and Table~\ref{tabcurvrot}):
\par\smallskip
\begin{figure}[ht]
\centering
\begin{tikzpicture}[domain=0:10,yscale=.6,xscale=.8]
	\draw[->] (-0.2,0.15) -- (10.2,0.15) node[right] {\large$r$};
	\draw[->] (0,-.2) -- (0,10.2) node[left] {$v$};
	\draw[color=black] plot[id=min,samples=300] function{10.8*sqrt((x*x*x < 1) ? x*x*x : 1)*sqrt((.3+.5*x)/x)} node[right] {$v_1(r)$};
	\draw[color=blue] plot[id=gauss,samples=100] function{7.2*sqrt(sqrt(pi)*erf(x)-2*x*exp(-x**2))*sqrt((.3+.5*x)/x)} node[right] {$v_2(r)$};

	\draw[color=green] plot[id=lorentz,samples=100] function{3.3*sqrt(x*sqrt(2)-atan(x*sqrt(2)))*sqrt((.3+.5*x)/x)} node[right] {$v_3(r)$};
	\draw[color=purple] plot[id=exp,samples=100] function{6.2*sqrt(2-exp(-4*x)*(16*x**2+8*x+2))*sqrt((.3+2*x)/(4*x))} node[right] {$v_4(r)$};
	\draw[color=brown] plot[id=poiss,samples=100] function{3.1*sqrt(6-exp(-3*x)*(27*x**3+27*x**2+18*x+6))*sqrt((.3+1.5*x)/(3*x))} node[right] {$v_5(r)$};
	\draw[color=violet] plot[id=mix,samples=300,domain=0:2.6,xscale=4,yscale=1.1] function{
	10*sqrt((1000*x*x*x < 1) ? 1000*x*x*x : 1)*sqrt((.3+5*x)/(10*x))
	+sqrt(6-exp(-3*x)*(27*x**3+27*x**2+18*x+6))*sqrt((.3+1.5*x)/(3*x))
	} node[right] at(10,10.2) {$v_6(r)$};
\end{tikzpicture}
\caption{\label{curverot}}
\end{figure}
\begin{table}[ht]
\centering
\begin{tabular}{clcc}
\hline\hline
Rotational velocity& Mass distribution&$a$&$b$\\ 
\hline 
$v_1$&\rule{0pt}{24pt}$
\rho_1(r)=\begin{cases}
3&\text{if $r\leq1$}\\
0&\text{if $r\geq1$}
\end{cases}$   &$-0.7$&$0.5$\\ 
$v_2$& $\rho_2(r)=e^{-r^2}$ (Gaussian)&$-0.7$&$0.5$\\ 
$v_3$&  $\rho_3(r)=\frac{1}{1+2\,r^2}$ (Lorentzian)&$-0.7$&$0.5$\\ 
$v_4$& $\rho_4(r)=e^{-r}$ (Exponential) &$-0.7$&$0.5$\\ 
$v_5$& $\rho_5(r)=r\,e^{-r}$ &$-0.7$&$0.5$\\ 
$v_6$& $\rho_6(r)=10000\,\rho_1(10\,r)+\rho_5(r)$ &$-0.7$&$5$\\ 
\hline 
\end{tabular}
\caption{\label{tabcurvrot}}
\end{table} 
Then, we see that the curves have a resemblance to the observed data.
This suggests more exhaustive studies of the relations between galactic mass densities and rotational velocity fields according to relations \eqref{rvf} with varying exponent $k$\,.
Moreover, the modified force $\vec{F}$ in \eqref{mondnew} depends on the time $t$ which could be related to a notion of cosmological expansion; a relation which could  be also studied in future works.
\par
Finally, a last question rises from these projective aspects: what could be the \textit{vanishing points} in such spacetime manifolds modeled locally by four-dimensional projective spaces? These vanishing points are at infinity in a projective space of dimension 3, and then, they appear to be points of congruences of timelike worldlines not necessarily crossing in the four-dimensional spacetime.  Hence, could this produce a sort of Big-Bang effect?
\section*{Conflicts of interest}
The author declares that there is no conflict of interest regarding the publication of this article.
\section*{Acknowledgments}
The author would like to thank Professor Ralf Hofmann from \emph{Karlsruhe Institute of Technology} (\mbox{Karlsruhe}, Germany), and Doctor Thierry Grandou from \emph{Institut de Physique de Nice} (Valbonne, France) for their invitation to participate at the `\emph{5th Winter Workshop on Non-Perturbative Quantum Field Theory}' which was held at Sophia-Antipolis (22-24 march 2017, France) to present the principles of relativistic localizing systems.
\bibliographystyle{ieeetr} %plain abbrv unsrt siam apalike acm ieeetr

\begin{thebibliography}{10}

\bibitem{Rubin2015}
J.~L. Rubin, ``Relativistic pentametric coordinates from relativistic
  localizing systems and the projective geometry of the spacetime manifold,''
  {\em Electronic Journal of Theoretical Physics}, vol.~12, no.~32,
  pp.~83--112, 2015.

\bibitem{Coll:2000}
B.~Coll, ``Elements for a {T}heory of {R}elativistic {C}oordinate {S}ystems:
  {F}ormal and {P}hysical {A}spects,'' in {\em Reference Frames and
  Gravitomagnetism} (J.-F. Pascual-S{\'a}nchez, L.~Flor{\'{\i}}a, A.~S. Miguel,
  and F.~Vicente, eds.), Proceedings of the XXIII Spanish Relativity Meeting
  (EREs2000), pp.~53--65, World Scientific Publishing Company, Incorporated,
  (Valladolid, 6--9 September 2000) 2001.

\bibitem{bahder2001navigation}
T.~B. Bahder, ``Navigation in curved space--time,'' {\em American Journal of
  Mathematics}, vol.~69, no.~3, pp.~315--321, 2001.

\bibitem{Blago:2002}
M.~Blagojevi{\'c}, J.~Garecki, F.~W. Hehl, and Y.~N. Obukhov, ``Real null
  coframes in general relativity and {GPS} type coordinates,'' {\em Physical
  Review D}, vol.~65, no.~4, p.~044018(6), 2002.

\bibitem{Rovelli:2002gps}
C.~Rovelli, ``{GPS} observables in general relativity,'' {\em Physical Review
  D}, vol.~65, no.~4, p.~044017(6), 2002.

\bibitem{Coll2003jsrs}
B.~Coll, ``A principal positioning system for the {E}arth,'' in {\em
  Journ{\'e}es 2002 -- syst{\`e}mes de r{\'e}f{\'e}rence spatio-temporels.
  Astrometry from ground and from space} (N.~{Capitaine} and M.~{Stavinschi},
  eds.), vol.~14, pp.~34--38, 2003.

\bibitem{coll2006two}
B.~Coll, J.~J. Ferrando, and J.~A. Morales, ``Two-dimensional approach to
  relativistic positioning systems,'' {\em Physical Review D}, vol.~73, no.~8,
  p.~084017, 2006.

\bibitem{Coll2006aipconf}
B.~Coll, ``Relativistic positioning systems,'' in {\em AIP Conference
  Proceedings}, vol.~841, pp.~277--284, 2006.

\bibitem{Coll2012esawrk}
B.~Coll, ``Relativistic positioning systems: Perspectives and prospects,'' in
  {\em Relativistic Positioning Systems and their Scientific Applications},
  (Brdo (Slovenia)), \emph{ESA Advanced Concepts Team} and the University of
  Ljubljana, 19th to 21th September 2012.

\bibitem{EhlPiSch:72}
J.~Ehlers, F.~A.~E. Pirani, and A.~Schild, ``The geometry of free fall and
  light propagation,'' in {\em General relativity, papers in honour of J.L.
  Synge} (L.~O{'}Raifeartaigh, ed.), (Oxford), pp.~63--84, Clarendon Press,
  1972.

\bibitem{Woodhouse73}
N.~M.~J. Woodhouse, ``The differentiable and causal structures of space-time,''
  {\em Journal of Mathematical Physics}, vol.~14, no.~4, pp.~493--501, 1973.

\bibitem{Kingal76}
A.~R. King, S.~W. Hawking, and P.~J. McCarthy, ``A new topology for curved
  space-time which incorporates the causal, differential, and conformal
  structures,'' {\em Journal of Mathematical Physics}, vol.~17, no.~2,
  pp.~174--181, 1976.

\bibitem{Malament77}
D.~B. Malament, ``The class of continuous timelike curves determines the
  topology of spacetime,'' {\em Journal of Mathematical Physics}, vol.~18,
  no.~7, pp.~1399--1404, 1977.

\bibitem{Seifert1980}
H.~Seifert and W.~Threlfall, {\em Seifert and Threlfall: A Textbook of
  Topology}, vol.~89 of {\em Series: Pure and Applied Mathematics}.
\newblock Academic Press, 1980.

\bibitem{Garcia-Parrado2005}
A.~Garc{\'i}a-Parrado and J.~M.~M. Senovilla, ``Causal structures and causal
  boundaries,'' {\em Class. Quantum Grav.}, vol.~22, pp.~R01--R84, 2005.

\bibitem{Kronheimer1967}
E.~H. Kronheimer and R.~Penrose, ``On the structure of causal spaces,'' {\em
  Proceedings of the Cambridge Philosophical Society}, vol.~63, pp.~481--501,
  1967.

\bibitem{ehres49}
C.~Ehresmann, ``Les connexions infinit{\'e}simales dans un espace fibr{\'e}
  diff{\'e}rentiable.'' Centre Belge Rech. Math., Colloque de Topologie,
  Bruxelles, 5th to 8th june 1950, 29--55, 1951.

\bibitem{Famaey2012}
B.~Famaey and S.~S. McGaugh, ``Modified Newtonian dynamics (MOND):
  Observational phenomenology and relativistic extensions,'' {\em Living
  Reviews in Relativity}, vol.~15, no.~1, p.~10, 2012.

\bibitem{Cartan1981}
{\'E}.~Cartan, {\em The Theory of Spinors}.
\newblock New York: Dover Publications, Inc., 1981.

\end{thebibliography}
\end{document}